\documentclass[journal,11pt]{article}

\usepackage{graphicx}
\usepackage{amssymb}
\usepackage{epstopdf}
\usepackage{amsmath,amsthm}
\usepackage{enumerate}
\usepackage{eufrak}
\usepackage{cite}
\usepackage{mathcomp}
\usepackage{supertabular}
\usepackage{longtable}
\usepackage{stmaryrd}
\usepackage{url}
\usepackage{color}
\usepackage{rotating}
\usepackage{float}
\usepackage[ruled,vlined, linesnumbered]{algorithm2e}

\usepackage[normalem]{ulem}
\usepackage[makeroom]{cancel}
 \usepackage[subnum]{cases}

\usepackage{fullpage}

\usepackage{enumerate}
\usepackage[all,cmtip]{xy}

\newcommand{\be}{\begin{equation}}
\newcommand{\ee}{\end{equation}}
\newcommand{\bear}{\begin{eqnarray}}
\newcommand{\eear}{\end{eqnarray}}
\newcommand{\bears}{\begin{eqnarray*}}
\newcommand{\eears}{\end{eqnarray*}}
\newcommand{\bi}{\begin{itemize}}
\newcommand{\ei}{\end{itemize}}
\newcommand{\ben}{\begin{enumerate}}
\newcommand{\een}{\end{enumerate}}

\newcommand{\ceiling}[1]{\left\lceil #1\right\rceil}

\newcommand{\vA}{\mathbf A}

\newcommand{\vM}{\mathbf M}

\newcommand{\vC}{\mathbf C}
\newcommand{\vD}{\mathbf D}
\newcommand{\vV}{\mathbf V}

\newcommand{\vI}{\mathbf I}

\newcommand{\vx}{\mathbf x}
\newcommand{\vdd}{\mathbf d}
\newcommand{\vu}{{\mathbf u}^{(1)}}
\newcommand{\vv}{{\mathbf u}^{(2)}}
\newcommand{\vuu}{{\mathbf v}^{(1)}}
\newcommand{\vvv}{{\mathbf v}^{(2)}}

\newcommand{\vhx}{\hat{\mathbf x}}

\newcommand{\vtx}{\tilde{\mathbf x}}
\newcommand{\vty}{\tilde{\mathbf y}}
\newcommand{\vtu}{\tilde{\mathbf u}^{(1)}}
\newcommand{\vtv}{\tilde{\mathbf u}^{(2)}}
\newcommand{\vhu}{\hat{\mathbf u}^{(1)}}
\newcommand{\vhv}{\hat{\mathbf u}^{(2)}}

\newcommand{\vtA}{\tilde{\mathbf A}}

\newcommand{\vy}{\mathbf y}
\newcommand{\vc}{\mathbf c}
\newcommand{\vd}{\mathbf d}
\newcommand{\ve}{\mathbf e}
\newcommand{\vzero}{\mathbf 0}

\newcommand{\imax}{i_{\rm max}}
\newcommand{\imin}{i_{\rm min}}
\newcommand{\prob}{{\rm Prob}}

\newcommand{\FF}{{\mathbb F}}

\newcommand{\enc}{\textsc{Encode}}
\newcommand{\rec}{\textsc{Reconstruct}}
\newcommand{\rep}{\textsc{Repair}}

\newcommand{\etal}[1]{{\em et al.}}

\newtheorem{theorem}{Theorem}
\newtheorem{lemma}[theorem]{Lemma}

\newtheorem{proposition}[theorem]{Proposition}

\theoremstyle{remark}
\newtheorem{remark}{Remark}
\newtheorem{example}[remark]{Example}

\title{Synchronizing Edits in Distributed Storage Networks}
\author{ 
	Salim El Rouayheb, 
	\and 
	Sreechakra Goparaju, 
	\and 
	Han Mao Kiah,
	\and 
	Olgica Milenkovic
}

\begin{document}
\maketitle

\begin{abstract}
We consider the problem of synchronizing data in distributed storage networks under an edit model that includes deletions and insertions. We present two modifications of MDS, regenerating and locally repairable codes that allow updates in the parity-check values to be performed with one round of communication at low bit rates and using small storage overhead. Our main contributions are novel protocols for synchronizing both hot and semi-static data and protocols for data deduplication applications, based on intermediary permutation, Vandermonde and Cauchy matrix coding. 
\end{abstract}

\section{Introduction}

Coding for distributed storage systems has garnered significant attention in the past few years~\cite{chang2008bigtable,shvachko2010hadoop,calder2011windows,zeng2009research}, due to the rapid development of information technologies and the emergence of Big Data formats that need to be saved or disseminated in a distributed fashion across a network. Reliable and efficient storage, transfer, retrieval and update of distributed data are challenging computational and networking tasks in which coding theory plays an indispensable role. One of the most significant applications of coding methods is to ensure data integrity under distributed component and node failure, via controlled redundancy introduced at different locations. This is accomplished by implementing modified Reed-Solomon, regenerating and local repair coding solutions~\cite{wicker1999reed,dimakis2010network,huang2012erasure,rashmi2009explicit, gopalan2012locality, gopalan2013explicit} that extend and generalize the traditional coding paradigms that are part of CDs, DVDs, flash memories and RAID like storage systems~\cite{immink2004codes,gregori2003chip,yaakobi2010error}. 

Typical distributed storage systems (DSS) such as Google's Bigtable, Microsoft Azure and the Hadoop Distributed File System, are all designed to scale to very large file sizes, to allow for low-latency data access and to store redundant coded copies of files over a set of servers, disks or nodes connected via a communication network. Two key functionalities of codes for DSS systems are reconstructability of files via access to a subset of the nodes; and content repairability of failed nodes that does not  compromise data reconstruction capabilities. Evidently, both these functionalities need to be retained when the content of files undergoes \emph{edits} which arise in dynamically changing storage systems and networked systems with load asymmetries involving ``hot data''~\cite{weil2006ceph} or data deduplication features~\cite{meister2009multi}. Other examples in which frequent edits are encountered are Dropbox and Sugarsync shared file systems~\cite{drago2012inside}, in which large number of user files are being independently updated at multiple locations. 
Current solutions for synchronization protocols either assume that data in uncoded or do not fully exploit the distributed nature of information, while deduplication methods mostly apply to read-only architectures and are in an early stage of development as far as distributed systems are concerned~\cite{riteau2011shrinker,zhang2012droplet}.

When files are edited or deduplicated, the changes in the information content need to be communicated to the redundant storage nodes so that the DSS retains its reconstruction and repair capabilities, and so that the communication and update costs are minimized. This is a challenging task, as edits such as deletions and insertions, which most commonly arise in practice, cause synchronization issues that appear hard to account for in the DSS encoded domain. In Dropbox and related systems, deletion and insertion synchronization issues are resolved in the \emph{uncoded} domain via the use of the well known rsync~\cite{tridgell1996rsync} and dsync~\cite{knauth2013dsync} or zsync (http://zsync.moria.org.uk/) algorithms, related to a number of file synchronization methods put forward in the information theory literature~\cite{orlitsky2003one,minsky2003set,Venkataramanan.etal:2010,Ma.etal:2011, Yazdi.Dolecek:2012, Bitouze.Dolecek:2013,Venkataramanan.etal:2013, Yazdi.Dolecek:2014}. There, uncoded stored copies of a file are synchronized from an edited user's copy. More specifically, the problems involve a single user and a single node storing a replicate of the user's file. After the user edits his/her file, assuming {\em no knowledge} of the edits, the user and the storage node communicate interactively until their files are matched or until one node matches the master copy of the user. In the DSS scenario we propose to analyze, one may assume both {\em full knowledge} of the edits or unknown edits, pertaining to the scenarios that edits were made by the users themselves or by another user sharing the file, respectively. The core problem in this case is to efficiently update {\em coded} copies in storage nodes with minimal communication rate. It is worth pointing out that this type of question regarding synchronization from edits in DSS systems fundamentally differs from update efficient codes studied in~\cite{Anthapadmanabhan.etal:2010,Rawat.etal:2011,Mazumdar.etal:2012,Mazumdar.etal:2014}. In these contributions, the authors minimize the \emph{number of nodes} that need to be updated when a user's file is changed. Consequently, the edits in consideration may be vaguely viewed as substitutions, in which case, minimizing the communication cost reduces to minimizing the number of storage nodes a user needs to recruit during update. In contrast, 
this line of work is concerned with edits of the form of insertions or deletions,
where such simplified reductions do not apply. And although one may view a deletion as a sequence of substitution errors, using update efficient codes with such preprocessing is highly sub-optimal; in the worst case, one deletion is equivalent to $\ell$ substitutions, where $\ell$ is the file length. 
Furthermore, instead of minimizing the number of storage nodes a user needs to communicate, the objective of synchronizing from deletions and insertions is to minimize the communication cost between the users and storage node even at the cost of introducing a small, controlled amount of storage overhead. 

The contributions of the paper are three-fold. First, we describe a number of edit models and accompanying coding protocols for efficient synchronization in a distributed storage environment that maintains regenerative properties. The synchronization protocols are based on a simple new scheme termed \emph{intermediary coding}, which flexibly changes the structure of the code so as to allow reduced communication complexity between nodes in a distributed system. The intermediary coding schemes also offer flexibility in terms of \emph{accommodating a very broad family of coding schemes} used in storage devices (such as erasure codes, regenerating codes, locally repairable codes etc). Second, we study extensions of the edit models related to deduplication applications and different data types, such as hot and semi-static data. Third, we provide worst and average case communication cost analyses for different edit models and synchronization protocols. This analysis reveals that traditional schemes require a significantly higher communication cost than
schemes based on intermediary coding, both in the worst case and average case scenario. This may be attributed to the fact that traditional encoding requires each node to communicate symbols in the \emph{span} of all deletion positions in different nodes, while intermediary coding allows for reducing the communication cost to the number of bits needed to encode the particular edit only.

The paper is organized as follows. Section 2 introduces the relevant notation and provides the precise problem statement. Section 3 contains the description of simple update protocols and their underlying communication cost when traditional DSS encoding methods are used. Section 4 contains our main result, a collection of encoding algorithm and protocols for data synchronization that 
have overhead communication cost that is a constant factor away from the fundamental limits derived in Section 5. 
Section 5 then examines both storage overhead and communication overhead and explains how to trade between these two system parameters.
Section 6 provides a short discussion on how to handle unknown edit positions.
The average case analysis of the communication cost of traditional and intermediary coding schemes, the description of hybrid and deduplication schemes are presented in the Appendix section of the paper.

\section{Notation and Problem Statement}

Throughout the paper, we use $[n]$ to denote the set of integers $\{1,2,\ldots, n\}$. For a finite set $X$ and $k\le |X|$, the collection of $k$-subsets of $X$ is denoted by $\binom{X}{k}$. The symbol $\FF_q$ is reserved for a finite field of order $q$. The vector space of all vectors of length $\ell$, the vector space of all matrices of dimensions $m\times n$, and the vector space of all tensors of dimensions $m\times n\times \ell$
are denoted by $\FF_q^\ell$, $\FF_q^{m\times n}$ and $\FF_q^{m\times n\times \ell}$, respectively.
For $i\in [\ell]$, $x_i$ represents the $i$th coordinate of a vector $\vx\in\FF_q^\ell,$ while $\ve_i$ stands for the $i$th standard basis vector.

Given a matrix $\vM\in\FF_q^{m\times n}$ and a subset of the rows $R\subseteq [m]$ and a subset of the columns $C\subseteq[n]$, $\vM|_{R\times C}$ represents the $|R|\times |C|$ matrix obtained by restricting $\vM$ to rows in $R$ and columns in $C$. We use analogous definitions for tensors.

\vskip 5pt
 
 We provide next a straightforward example that explains the underlying motivation for the work as well as 
 the difficulties encountered in synchronizing coded data.
 
\begin{example}\label{exa:motivation}
Consider two users with data vectors 
$\vu=\left(u^{(1)}_1,u^{(1)}_2,u^{(1)}_3,u^{(1)}_4,u^{(1)}_5\right)$ and 
$\vv=\left(u^{(2)}_1,u^{(2)}_2,u^{(2)}_3,u^{(2)}_4,u^{(2)}_5\right)$, 
both of length five, over $\FF_q$. Here, $q$ is chosen based on the smallest number of consecutive, editable bits. Suppose we have three storage nodes where nodes 1, 2 and 3 store user information $\vu$ and $\vv$, and parity information 
$\vu+\vv$, respectively. Then the system is able to reconstruct both data blocks $\vu$ and $\vv$ and repair any failed node by accessing any two nodes.

Suppose both data blocks are subjected to a single symbol deletion, and that the resulting data blocks are 
$\vtu=\left(u^{(1)}_1,u^{(1)}_2,u^{(1)}_3,u^{(1)}_4\right)$ and $\vtv=\left(u^{(2)}_2,u^{(2)}_3,u^{(2)}_4,u^{(2)}_5\right)$. 
What protocol should the users employ and what information do they have to communicate to the three storage nodes so as to retain both reconstruction and repair functionalities, all with minimal data transmission cost?

One option is for the nodes to update their respective contents to $\vtu$, $\vtv$ and 
\[ \label{eq:nodes}
\vtu+\vtv=\left(u^{(1)}_1+u^{(2)}_2,u^{(1)}_2+u^{(2)}_3,u^{(1)}_3+u^{(2)}_4,u^{(1)}_4+u^{(2)}_5\right).
\]
Clearly, it is both necessary and sufficient for the user with data block $\vtu$ to communicate his/her deletion {\em position} to node 1.
A similar statement holds true for the user with data block $\vtv$ and node 2. But what is the minimum communication complexity needed for
node 3 to update its content to $\vtu+\vtv$? At first glance, it appears that both nodes 1 and 2 should transmit their whole content to node 3, as deletions occurred in their first and last position, respectively. The goal of this work is to show that, akin to the notion of functional repair in DSS \cite{dimakis2010network}, 
one may significantly save in communication complexity by updating the content of node 3 through a flexible change in the code structure.
\end{example}
\setcounter{remark}{0}
 
\subsection{Coding for Distributed Storage Systems}

Let $\vx=\left(x^{(1)}, x^{(2)}, \ldots, x^{(B)}\right) \in\FF_q^B$ be an information vector 
to be stored in a distributed storage system (DSS).
We call 
each $x^{(s)},$ for $s=1,\ldots,B$, a {\em data unit}.
A DSS is equipped with an {\em encoding} function, 
as well as a set of reconstruction and a set of repair algorithms specifying a coding scheme.
The encoding algorithm converts the vector $\vx$ into $n$ vectors of length $\alpha$
and stores them in $n$ {\em storage nodes}, while 
the reconstruction algorithm recovers $\vx$ from the contents of any $k\leq n$ out of $n$ nodes.
In addition, when a node fails, 
one uses the content of a subset of $d \leq n$ nodes to repair the content of the failed node.

More formally, we have the following definitions.

\begin{enumerate}
\item An {\em encoding} function is a map $\enc: \FF_q^B \to\FF_q^{n\times \alpha}$, where $\alpha$ is a given code parameter that 
has to be chosen so that the encoding function of interest is constructable. 

\item For any $T\in \binom{[n]}{k}$, a map $\rec(T):\FF_q^{k\times\alpha} \to\FF_q^{B}$
is termed a {\em reconstruction} function if for all $\left(x^{(1)}, x^{(2)}, \ldots, x^{(B)}\right)\in \FF_q^B$,
\begin{equation*}
\rec(T)\left.\left(\enc\left(x^{(1)}, x^{(2)}, \ldots, x^{(B)}\right)\right|_{T\times[\alpha]}\right)=\left(x^{(1)}, x^{(2)}, \ldots, x^{(B)}\right).
\end{equation*} 

\item Given a $t\in[n]$ and $T\in \binom{[n]\setminus\{t\}}{d}$, 
 a map $\rep(t,T):\FF_q^{d\times\alpha} \to\FF_q^{\alpha}$ 
 is termed an {\em exact repair} function
 if 
\begin{equation*}
\rep(t,T)(\vC|_{T\times [\alpha]})=\vC|_{\{t\}\times[\alpha]},
\end{equation*} 
\noindent where $\vC$ is the information stored over $n$ nodes.
\end{enumerate}

Depending on the set of subsets $T$ over which a repair function is defined on, one recovers the various descriptions of codes for DSSs studied in the literature:
\begin{enumerate}
\item {\em Maximum distance separable (MDS) codes}~\cite[Ch. 11]{macwilliams1977theory}. Here, we have $d=k=B$ and $\alpha=1$, and the reconstruction algorithm can also be used as the repair algorithm. We simply denote an MDS code of length $n$ and dimension $k$ by $[n,k]$, following the standard notation for the dimensions of a linear error-correcting code, where it is understood that the minimum Hamming distance of the code equals $n-k+1$. One of the key properties of MDS codes exploited for repair is the fact that every $k$-subset of the $n$ coordinates is an information set. 

\item {\em Regenerating codes}~\cite{WDR07,RSKR09}. Here, $d\ge k$, and we require a repair function $\rep(t,T)$
for all $t\in[n]$ and all $T\in \binom{[n]\setminus\{t\}}{d}$.
We note that in bandwidth-limited DSSs, it is of importance to define an intermediary function, mapping words from $\FF_q^{d\times\alpha}$ to $\FF_q^{\alpha}$. More precisely, a repair algorithm $\rep(t,T)$ maps inputs from $\FF_q^{d\times\alpha}$ to $\FF_q^{d\times \beta}$ and then to $\FF_q^{\alpha}$, where $\beta$ is usually smaller than $\alpha$ and indicates the required amount of downloaded information (bits).
However, as we observe later, this intermediary function is not needed in our subsequent analysis.

\item {\em Locally repairable codes (LRC)}~\cite{gopalan2012locality,gopalan2013explicit,RKSV13}. Here, $d<k$, and for all $t\in[n]$,
we only require a repair function $\rep(t,T)$ for {\em some} $T\in \binom{[n]\setminus\{t\}}{d}$.
Efficient repair is achieved by minimizing the number of nodes that needs to be contacted. 
In contrast to regenerating codes, we only require information from a particular set of $d$ nodes to repair a node,
and the amount of information downloaded is not considered.
\end{enumerate}
We broadly refer to the three aforementioned families of encodings for DSSs as $(n,k,d,\alpha,B)$ DSS codes, and focus on code maps that are linear. 

\subsection{Problem Description}

Assume next that the units of a data block are updated via deletion or insertion edits%
\footnote{Consider a string $(x_1,x_2,\ldots,x_\ell)$. A {\em deletion} at position $i$ 
results in $(x_1,x_2,\ldots,x_{i-1},x_{i+1},\ldots,x_\ell)$, while an {\em insertion} of symbol $a$ at position $i$ results in $(x_1,x_2,\ldots,x_{i-1},a,x_{i},\ldots,x_\ell)$.}.
For $s\in [B]$, consider the set 
 $N(s)\triangleq\{t\in[n]: \enc(\ve_s)|_{\{t\}\times[\alpha]}\ne \vzero\}$; in words, $N(s)$ denotes the set of nodes that needs to be updated when the data unit indexed by $s$ is edited. We say that the nodes in $N(s)$ are {\em connected to unit $s$}.
 
Given a $(n,k,d,\alpha,B)$ DSS code, one can extend it to a $(n,k,d,\alpha\ell,B\ell)$ DSS code for any choice of a positive integer $\ell$.
This extension is akin to data-striping in RAID systems, amounting to a simple conversion of a symbol to a vector of symbols.
One can also view this construction as a means of dividing the $B\ell$ data units into $\ell$ groups of $B$ symbols each 
and then applying to each group of symbols an $(n,k,d,\alpha,B)$ DSS encoding. 
Equivalently, in this $(n,k,d,\alpha\ell,B\ell)$ DSS coding scheme we may regard the given information as a $B$-tuple of {\em data blocks}, each of length $\ell$, 
say $\left(\vx^{(1)}, \vx^{(2)}, \ldots, \vx^{(B)}\right)$. 
The choice of $\ell$ is dictated by the storage system at hand, in particular, by the size of the block (which is a system parameter) and the expected size of an edit.
Consequently, we assume that for algorithmic approaches that $\ell$ is a fixed system parameter. Nevertheless, for asymptotic analysis, we use the standard modeling assumption in which $\ell$ is allowed to grow arbitrarily large. 

\begin{example}[Continued]
Recall the code of Example \ref{exa:motivation}. There, we tacitly considered a single parity $[3,2]$ MDS code, or a DSS code with parameters $n=3$, $k=2$, $d=2$, $\alpha=1$, $B=2$ and $\ell=5$, with codewords of the form 
$\left(u^{(1)},u^{(2)},u^{(1)}+u^{(2)}\right),$ and $u^{(1)},u^{(2)}$ belonging to a finite field.
Assume that there are two users with data blocks $\vu$ and $\vv$ of length $\ell$. 
Then by having nodes 1, 2 and 3 store $\vu,\vv,\vu+\vv$, we obtain 
a $(3,2,2,\ell,2\ell)$ DSS code. 
Furthermore, we observe that the connected nodes of the user with data block $\vu$ are nodes 1 and 3, 
while the connected nodes of the user with data block $\vv$ are nodes 2 and 3.
\end{example}
\setcounter{remark}{0}

The edit model of interest assumes that the data blocks are subjected to deletions performed in an independent fashion by $B$ different users%
\footnote{It may be possible that a user edits a number of different data blocks. 
However, for simplicity, our model assumes that each data block is edited by one user.}. 
More precisely, two different models for edits are studied, including:
\begin{enumerate}
\item {\em The uniform edits model,} in which each data block has the same number of deletions and 
thus the resulting data blocks all have the same length. Here the number of deletions is $o(\ell/\log \ell)$.
This model is used to describe the main ideas behind the 
work in a succinct and notationally simple manner, but may not be of practical importance. The model is also 
amenable for combinatorial analysis.
\item {\em The nonuniform edits model,} in which each data block has a possibly different number of edits. 
This model may be analyzed both in a combinatorial and probabilistic setting. 
In the former case, we assume $D \leq B \, \ell $ edits for all $B$ users. In the latter case, we assume that one is given the probability $p$ of deleting any particular symbol in a data block, resulting in an average number of $p \, \ell$ edits per data block.
Note that $p$ may depend on $\ell$.
\end{enumerate}

More generally, the problem of interest may be stated as follows: find the ``best'' protocol for the $B$ users to communicate their edits to the storage nodes,
so that the storage nodes can update their information while maintaining both reconstruction and repair functionalities.
The word ``best'' may refer to the smallest communication cost, smallest required combination of communication and additional storage cost, etc. For simplicity, we first focus on the case of smallest communication cost, defined as the average number of bits transmitted from one user to all the storage nodes. An obvious suboptimal approach is for each user to send the entire data file of length $\ell$ to each connected  storage node, so that each such node may update its information according to the encoding function at hand. However, as the number of edits made to a data block may be (and usually is) much smaller compared to the data block size, a communication cost of $\ell$ symbols or $\ell\log q$ bits may be highly suboptimal.

On the other hand, suppose a user performs a single edit. To encode the information about this edit, one requires 
$\log \ell$ bits for a deletion (for encoding the position) and $\log \ell+ \log q$ bits for an insertion (to encode the position and the symbol). Hence, assuming that the data blocks and edits are uncorrelated, one requires each user to communicate at least $\log \ell+ \log q$ bits for each edit. As a result, it is straightforward to see that a (loose) lower bound on the communication cost needed for synchronization from a \emph{constant number} of deletions is $\log \ell$ (deletions) or $\log \ell+ \log q$ (insertion) bits, as at least one user has to communicate the information about one of its edits. When the number of edits is of the order of the length of the data unit $\ell$, this bound may not apply, as it may be more efficient to communicate ``nonedits''. This difficult issue will not be discussed in this work, although all schemes proposed in this work apply to this case as well, but without a proof of order-optimality. More detailed descriptions and tighter lower bounds on the communication complexity are given in Section \ref{sec:fund+trade}.

In what follows, we propose a number of schemes that achieve a communication cost of $O(\log \ell+\log q)$ bits that is of the order of the intuitive lower bound. To facilitate such low communication cost, we introduce additional storage overhead needed to define an intermediary encoding function, which we refer to as intermediary coding. The gist of the encoding method is to transform the information, and hence the codes applied to data blocks, via permutation, Vandermonde and Cauchy matrices (the resulting schemes are subsequently referred to as Scheme P, V and C, respectively). The key property of the transforms is that they reduce the update and synchronization communication cost by changing the code structure. 
The storage overhead induced by this codes is carefully controlled by choosing the parameters of the corresponding matrices, as described in Section \ref{sec:tradeoff}.
In Section \ref{sec:tradeoff}, we also demonstrate that our schemes are optimal in terms of storage allocation when the number of edits is $o(\ell/\log \ell)$. Under the same condition of $o(\ell/\log \ell)$ edits, we demonstrate in Appendix \ref{sec:prob} that on the average our schemes outperform
schemes that do not utilize intermediary encoding functions.

In general, our derivations and methods do not rely on specific assumptions on the network topology. 
For the proofs of lower bounds or fundamental limits, a user is allowed to communicate with any other user or storage node and vice versa, i.e., 
the storage network is a complete graph. However, users are naturally assumed to communicate only with a ``minimal'' set of storage nodes which contains encodings or systematic repetition of their data block symbols.
These assumptions are summarized in Table \ref{tab:network}.

{
\begin{table}[t]
\centering
\begin{tabular}{|p{4cm}|p{7cm}|c|} 
\hline
& Assumptions & Schematic Diagram \\ \hline
Topology used for proving fundamental performance bounds. &
A user is allowed to communicate with any other user or storage node and vice versa. &
 \xymatrix{
 \bigcirc \ar@{<->} [r] \ar@{<->}[dr]\ar@{<->}[ddr] \ar@{<->}[d]& \Box\ar@{<->} [d] \ar@/^/@{<->}[dd]\\
 \bigcirc \ar@{<->} [r] \ar@{<->}[ur]\ar@{<->}[dr] & \Box\ar@{<->} [d] \\
  & \Box\\
  } \\
  \hline
Traditional Coding Scheme (See Scheme T) &
A user is allowed to communicate to any of his/her connected storage nodes and all users are allowed to have a two-way communication with a designated central node tasked with computing the span of all the deletions. &
 \xymatrix{
 \bigcirc \ar@{->} [r] \ar@{<->}[ddr] & \Box \\
 \bigcirc \ar@{->} [r] \ar@{<->}[dr] & \Box\\
  & \triangle\\
  } \\
  \hline
Intermediary Coding Schemes (See Scheme P and V)&
A user is allowed to communicate to any of his/her neighboring storage nodes. &
 \xymatrix{
 \bigcirc \ar@{->} [r] \ar@{->}[ddr] & \Box \\
 \bigcirc \ar@{->} [r] \ar@{->}[dr] & \Box\\
  & \Box\\
  } \\
  \hline
\end{tabular} 
\vskip 5pt
Here, $\bigcirc$ denotes a user, while $\Box$ and $\triangle$ denote a node and a designated central node, respectively.

\caption{Communication Network Topology Models.}
\label{tab:network}
\end{table} 
}

\section{Communication Cost with Traditional Encoding Schemes}
\label{sec:traditional}

As discussed earlier, assume that one is given an $(n,k,d,\alpha,B)$ DSS code over $\FF_q$ with functions $\enc$, $\rec$ and $\rep$,
and that for a given $\ell$, this code is augmented to an $(n,k,d,\alpha\ell,B\ell)$ DSS code. 
Next, we state explicitly the encoding, reconstruction and repair functions for this scheme and examine the requisite communication between users and storage nodes under a single deletion model. The uniform deletion model may be analyzed in the same manner. We first show that in the worst case, the number of bits communicated is at least $(\ell-1)$ symbols, or $(\ell-1)\log q$ bits; we then proceed to introduce a scheme that lowers this communication cost to an order-optimal level.

More formally, for a fixed $\ell$, in order to construct an $(n,k,d,\alpha\ell,B\ell)$ DSS code, 
we extend $\enc$ to the linear map%
\footnote{For compactness, unless stated otherwise, we drop the transposition symbol $T$ when writing the $B\times\ell$ array representing the $B$ data blocks of length $\ell$. In other words, we write $\left(\vx^{(1)}, \vx^{(2)}, \ldots, \vx^{(B)}\right)$ instead of $\left(\vx^{(1)}, \vx^{(2)}, \ldots, \vx^{(B)}\right)^T$.} 
$\enc(\ell):\FF_q^{B\times\ell} \to \FF_q^{n\times\alpha\times \ell}$, 
such that for $i\in[\ell]$,
\begin{equation}\label{eq:enc}
\enc(\ell)\left.\left(\vx^{(1)}, \vx^{(2)}, \ldots, \vx^{(B)}\right)\right|_{[n]\times [\alpha]\times \{i\}}=
\enc \left(x^{(1)}_i, x^{(2)}_i, \ldots, x^{(B)}_i\right).
\end{equation}

\noindent Hence, we regard the information stored at the $n$ nodes 
as an $n\times\alpha\times \ell$ tensor. For $i\in[\ell]$, 
the $i$th slice of this tensor is the $n\times\alpha$ array
$\enc \left(x^{(1)}_i, x^{(2)}_i, \ldots, x^{(B)}_i\right)$.
We refer to this particular encoding $\enc(\ell)$ as a {\em traditional encoding} function.
In the subsequent sections, we define another encoding function that uses $\enc(\ell)$ as a building block, 
and this novel encoding function will represent a crucial component of our low-communication cost synchronization schemes.

Next, we provide the corresponding reconstruction and repair functions accompanying traditional encoding and 
verify that we indeed have a DSS code. 

In particular, we define 
$\rec(T;\ell):\FF_q^{k\times\alpha\times \ell} \to\FF_q^{B\times\ell}$ for $T\in \binom{[n]}{k}$
and repair function $\rep(t,T;\ell):\FF_q^{d\times\alpha\times \ell } \to\FF_q^{\alpha\times\ell}$
for $t\in[n]$ and $T\in \binom{[n]\setminus\{t\}}{d}$ via 
\begin{align}
\left.\rec(T;\ell)\left(\vC|_{T\times [\alpha]\times[\ell]}\right)\right|_{[B]\times\{i\}}
&\triangleq \rec(T)\left(\vC|_{T\times[\alpha]\times\{i\}}\right)
& \mbox{ for }i\in[\ell]\label{eq:rec},\\
\left.\rep(t,T;\ell)\left(\vC|_{T\times [\alpha]\times[\ell]}\right)\right|_{[\alpha]\times\{i\}}
&\triangleq \rep(t,T)\left(\vC|_{T\times[\alpha]\times\{i\}}\right)
& \mbox{ for }i\in[\ell]\label{eq:rep}.
\end{align} 

\noindent Given access to $k$ nodes, we regard their information as a $k\times\alpha\times\ell$ tensor and apply the classical reconstruction algorithm to the $i$th slice for $i\in[\ell]$. 
Since this slice is in fact $\enc\left(x^{(1)}_i, x^{(2)}_i, \ldots, x^{(B)}_i\right)$, 
we retrieve the data units $\left(x^{(1)}_i, x^{(2)}_i, \ldots, x^{(B)}_i\right)$,
which correspond to the $i$th coordinates of the data blocks $\vx^{(1)},\vx^{(2)},\ldots, \vx^{(B)}$.
Hence, the algorithm given by \eqref{eq:rec} retrieves the complete collection of data blocks of the $B$ users. 
A similar setup holds for the repair algorithm. As a result, $\enc(\ell)$, $\rec(T;\ell)$ and $\rep(t,T;\ell)$ describe an $(n,k,d,\alpha\ell, B\ell)$ DSS code.

\vskip 5pt
 
Suppose after the edits, the data blocks are updated to $\vtx^{(1)}, \vtx^{(2)}, \ldots, \vtx^{(B)}$, each of length $\ell'$.
One straightforward approach to maintain reconstruction and repair properties  after the updates
is to ensure that the information stored at the $n$ nodes is given by $\enc(\ell')\left(\vtx^{(1)}, \vtx^{(2)}, \ldots, \vtx^{(B)}\right)$.
Unfortunately, as we will see in the next example, such an approach requires a communication cost of at least $(\ell-1)\log q$ bits in the worst case scenario.

\begin{example}[Generalized]\label{exa:bad}
Recall the systematic single parity $[3,2]$ MDS code and let the data blocks of users 1 and 2 be $\vu,\vv\in\FF_q^\ell$, respectively.
For this code,  we may use \eqref{eq:enc} and rewrite the information at the storage nodes as
\be \label{eq:bad}
\enc(\ell)\left(\vu,\vv\right)=\left(\begin{array}{c}\vu\\ \vv\\ \vu+\vv\end{array}\right).
\ee
\end{example}

Recall that in the case for MDS codes we have $\alpha=1$. 
Hence, the image of the encoding map is simply written as a matrix, instead of a three dimensional tensor. This notational convention applies to all subsequent examples in the paper.

\begin{proposition}\label{claim:trivial}
Let $\vu$ and $\vv$ be strings of length $\ell$.
Assume that exactly one deletion has occurred in $\vu$ and $\vv$.
Let $\vtu$ and $\vtv$ denote the respective edited strings, and 
suppose that the information to be updated at the storage nodes is given by \eqref{eq:bad}.
Then, in the worst case over all possible edit locations, 
the total communication cost is at least $\ell-1$ symbols, or $(\ell-1)\log q$ bits, 
independent of the network topology between users and storage nodes.
\end{proposition}

\begin{proof}
To prove the claimed result, we adapt the {\em fooling set method} from communication complexity (see \cite[\S 13.2.1]{AroraBarak:2009}), 
standardly used to lower bound the deterministic communication complexity of a function.

Assume that $u^{(1)}_\ell$ and $u^{(2)}_1$ are the deleted coordinates. Then the task of node 3 is to update its value to
$\left(u^{(1)}_1+u^{(2)}_2,u^{(1)}_2+u^{(2)}_3,\ldots, u^{(1)}_{\ell-1}+u^{(2)}_\ell\right)$.

Let $f:\FF_q^\ell\times \FF_q^\ell\to \FF_q^{\ell-1}$ be a function with 
$f(\vu,\vv)=\left(u^{(1)}_1+u^{(2)}_2,u^{(1)}_2+u^{(2)}_3,\ldots, u^{(1)}_{\ell-1}+u^{(2)}_\ell\right)$. 
A {\em fooling set} for $f$ of size $M$ is a subset $S \subseteq \FF_q^\ell\times \FF_q^\ell$ and 
and a value $\vc \in \FF_q^{\ell-1}$ such that 
(a) for every $\left(\vu, \vv\right) \in S$, $f\left(\vu, \vv\right) = \vc$ and (b) for every distinct 
$\left(\vu, \vv\right), \left(\vuu, \vvv\right)\in S$, 
either $f(\vu, \vvv)\ne \vc$ or $f(\vuu, \vv)\ne\vc$.
One can show that if $f$ has a fooling set of size $M$, 
then the total deterministic communication cost for any protocol computing $f$ is at least $\log M$.

To prove the claim, we exhibit next a fooling set of size $q^{\ell-1}$ for the function of interest.
Consider the subset of $\FF_q^\ell\times \FF_q^\ell$ of size $q^{\ell-1}$ defined as 
\[S=\left\{\left(\vu,\vv\right):\vu=\left(0,u^{(1)}_2,u^{(1)}_3,\ldots,u^{(1)}_{\ell}\right),\vv= \left(-u^{(1)}_2,-u^{(1)}_3,\ldots,-u^{(1)}_{\ell},0 \right)\right\},\]
 and let $\vc=\vzero$. 
Then $f \left(\vu,\vv\right)=\vzero$ for all $(\vu, \vv) \in S$. Furthermore, if $\vu\ne \vuu$,  or 
equivalently, if $\left(0,u^{(1)}_2,u^{(1)}_3,\ldots,u^{(1)}_{\ell}\right)\ne \left(0,v^{(1)}_2,v^{(1)}_3,\ldots,v^{(1)}_{\ell}\right)$,
then we can check that $f \left(\vu,\vvv\right)= \left(u^{(1)}_2-v^{(1)}_2,u^{(1)}_3-v^{(1)}_3, \ldots, u^{(1)}_\ell-v^{(1)}_\ell\right)\ne \vzero$. 
Therefore, $S$ is a fooling set of size $q^{\ell-1}$.
\end{proof}

For the general case involving more than two users and more than three storage nodes, one 
can focus on the worst case scenario in which two users each have a single deletion and need to update a parity check value in a common, connected node. The proof of Proposition \ref{claim:trivial} 
can be easily modified to show that in this case, the worst case communication cost remains $(\ell-1)$ symbols, or $(\ell-1)\log q$ bits.


\subsection{Update Protocols: Beyond the Worst Case}

In what follows, we describe a straightforward update protocol for the traditional encoding scheme with edits that are not necessarily confined to the worst case configuration.

Consider a $(n,k,d,\alpha\ell,B\ell)$ DSS code with 
the data blocks of the users equal to $\vx^{(1)}, \vx^{(2)}, \ldots, \vx^{(B)}$, and the information stored in the storage nodes equal to $\enc(\ell)\left(\vx^{(1)}, \vx^{(2)}, \ldots, \vx^{(B)}\right)$.
We consider the uniform edits model and for simplicity, 
assume that there is a single deletion%
\footnote{Insertions can be treated in an almost identical manner as deletions and will not be explicitly discussed.} 
at coordinate $i_s$ in data block $\vx^{(s)}$, 
for $s\in [B]$. Hence, the updated data block length equals $\ell'=\ell-1$.
Let $\vtx^{(1)}, \vtx^{(2)}, \ldots, \vtx^{(B)}$ be the edited data blocks.
In the traditional encoding scheme, to preserve reconstruction and repair properties,
we require the information stored by the nodes to be updated to 
$\enc(\ell-1)\left(\vtx^{(1)}, \vtx^{(2)}, \ldots, \vtx^{(B)}\right)$.

Example \ref{exa:bad} demonstrates that in the worst case, a user needs to transmit $(\ell-1)\log q$ symbols to its connected nodes.
Clearly, for most edit scenarios, certain portions of the encoded data do not need to be updated. 
Continuing Example \ref{exa:bad}, suppose 
$\vtx=(x_1,x_2,\ldots,x_{\ell-2},x_{\ell})$ 
with deletion at coordinate $\ell-1$ and 
$\vty=(y_1,y_2,\ldots,y_{\ell-2},y_{\ell-1})$ 
with deletion at coordinate $\ell$.
Node 3 needs to update its string to
$\vtx+\vty=(x_1+y_1,x_2+y_2,\ldots,x_{\ell-2}+y_{\ell-2},x_{\ell}+y_{\ell-1})$, and 
it suffices for users 1 and 2 to transmit $x_{\ell}$ and $y_{\ell-1}$, 
respectively, to node 3.
Indeed, to compute $\vtx+\vty$, node 3 only needs to compute the last coordinate 
as the other coordinates remain the same.
\vskip 5pt

Suppose that all users send their deleted coordinates to a designated central storage node, 
and that the designated central node computes $\imax\triangleq\max_{s\in [B]}i_s$ and $\imin\triangleq\min_{s\in [B]}i_s$.

Define  $I=\{i\in[\ell]: \imin\le i\le \imax-1\}$. 
Since $\left(\tilde{x}^{(1)}_i,\tilde{x}^{(2)}_i,\ldots,\tilde{x}^{(B)}_i\right)=\left({x}^{(1)}_i,{x}^{(2)}_i,\ldots,{x}^{(B)}_i\right)$ for $i< \imin$, and 
$\left(\tilde{x}^{(1)}_i,\tilde{x}^{(2)}_i,\ldots,\tilde{x}^{(B)}_i\right)=\left({x}^{(1)}_{i+1},{x}^{(2)}_{i+1},\ldots,{x}^{(B)}_{i+1}\right)$ for $i\ge \imax$, we have
\begin{align}
&\enc(\ell-1)\left.\left(\vtx^{(1)}, \vtx^{(2)}, \ldots, \vtx^{(B)}\right)\right|_{[n]\times[\alpha]\times\{i\}} \notag \\
&~~~~~~~~~~~~~~~~~~~~~~~~~~~=
\begin{cases}
\enc(\ell)\left.\left(\vx^{(1)}, \vx^{(2)}, \ldots, \vx^{(B)}\right)\right|_{[n]\times[\alpha]\times\{i\}},
& \mbox{if $i<\imin$},\\
\enc(\ell)\left.\left(\vx^{(1)}, \vx^{(2)}, \ldots, \vx^{(B)}\right)\right|_{[n]\times[\alpha]\times\{i+1\}},
& \mbox{if $i\ge\imax$.}
\end{cases}\label{eq:enct1}
\end{align}

%
%
%
%
%
%
%
%
%

In other words, the information stored at coordinates in $[n]\times [\alpha]\times ([\ell']\setminus I)$ need not be updated. Therefore, it suffices to compute $\enc(\ell-1)\left.\left(\vtx^{(1)}, \vtx^{(2)}, \ldots, \vtx^{(B)}\right)\right|_{[n]\times[\alpha]\times I}$ as given by 

\begin{equation}\label{eq:enct2}
\enc(\ell-1)\left.\left(\vtx^{(1)}, \vtx^{(2)}, \ldots, \vtx^{(B)}\right)\right|_{[n]\times[\alpha]\times I}
=\enc(|I|)\left(\vtx^{(1)}_{I}, \vtx^{(2)}_{I}, \ldots, \vtx^{(B)}_{I}\right),
\end{equation}
\noindent where $\vtx^{(s)}_I$ is the updated string restricted to the coordinates in $I$.
In other words, 
$\vtx^{(s)}_I\triangleq \left(x^{(s)}_{\imin},x^{(s)}_{\imin+1},\ldots, x^{(s)}_{i_s-1},x^{(s)}_{i_s+1},\ldots, x^{(s)}_{\imax}\right)$.

Consequently, equipped with $\vtx^{(s)}_I$ for $s\in[B]$ and equations \eqref{eq:enct1} and \eqref{eq:enct2}, 
the nodes are able 
to compute $\enc(\ell-1)\left(\vtx^{(1)}, \vtx^{(2)}, \ldots, \vtx^{(B)}\right)$.
We call this update protocol {\em Scheme T}.

%
%
%
%

\begin{proposition}[Scheme T]\label{prop:schemet}
Consider an $(n,k,d,\alpha\ell,B\ell)$ DSS code and assume single deletions in the user data blocks. 
The updates in accordance to Scheme T result in an $(n,k,d,\alpha(\ell-1), B(\ell-1))$ DSS code.
Each user needs to send out $|I|\log q$ bits to a connected storage node, which in the worst case equals $(\ell-1)\log q$ bits.
Here, $I=\{i\in[\ell]: \imin\le i\le \imax-1\}$, where $\imax=\max_{s\in [B]}i_s$ and $\imin=\min_{s\in [B]}i_s$.
\end{proposition}
In Appendix \ref{sec:prob}, we compute the expected communication cost between a user and a connected storage node for a number of probabilistic edit models, and show that for all models considered this cost is of the order of $\ell \log q$ bits, i.e., of the order of the worst case communication scenario rate. The analysis is based on order statistics that characterizes the average span of deletions in the data blocks.
%
%

\section{Synchronization Schemes with Order-Optimal Communication Cost}\label{sec:schemes}

To avoid repeated transmissions of all data blocks of users, one needs to develop encoding methods that work around the problems associated with the simple protocol T. We present next our main results, a collection of synchronization schemes that 
achieve communication cost of $O(\log \ell+\log q)$ bits as opposed to $O(\ell\log q)$ 
required in the T protocol setting. Our schemes are based on the following observation, apparent from Example \ref{exa:bad}: if one insists on using \eqref{eq:enc} to encode data, then edit updates inevitably incur a communication cost of the order of the data block lengths. Hence, the idea is to introduce an intermediary 
encoding algorithm that preserves the reconstruction and repair properties of the DSS code, while allowing for significantly lower edit information communication costs. This intermediary encoding algorithm is described in what follows and is illustrated in Figure~\ref{fig:encoding}.

\begin{figure}
\[
\small
\xymatrix@C=50mm@R=0mm{
*+<15mm>[F]{\left(\vx^{(1)}, \vx^{(2)}, \ldots, \vx^{(B)}\right)} 
\ar@<5mm>^(0.55){\enc(\ell)}[r]  &
*+<10mm>[F]{\begin{array}{c}\vC\\\\\vC|_{T\times[\alpha]\times[\ell]}\end{array}}
\ar@<5mm>^(0.45){\rec(T;\ell)}[l] \\
\txt{User's data blocks}&
\txt{Nodes}
}
\]
\[
\footnotesize
\xymatrix@C=25mm@R=0mm{
*+<12mm>[F]{\left(\vx^{(1)}, \vx^{(2)}, \ldots, \vx^{(B)}\right)} 
\ar@<2mm>[r] 
\ar@<8mm>^(0.5){\enc^*(\ell)}[rr] &
*+<4mm>[F.]{\left(\vx^{(1)}\vA^{(1)}, \vx^{(2)}\vA^{(2)}, \ldots, \vx^{(B)}\vA^{(B)}\right)} 
\ar@<2mm>^(0.62){\enc(\ell)}[r]
\ar@<2mm>[l]  &
*+<8mm>[F]{\begin{array}{c}\vC\\\\\vC|_{T\times[\alpha]\times[\ell]}\end{array}}
\ar@<2mm>^(0.38){\rec(T;\ell)}[l] 
\ar@<8mm>^(0.5){\rec^*(T;\ell)}[ll]\\
\txt{User's data blocks}&
&
\txt{Nodes}
}
\]

\caption{a) A diagram illustrating how traditional encoding extends an existing $(n,k,d,\alpha, B)$ DSS code to 
an  $(n,k,d,\alpha\ell, B\ell)$ DSS code so as to handle data blocks of length $\ell$. b) Intermediary encoding that retains both reconstruction and repair functionalities. }
\label{fig:encoding}
\end{figure}
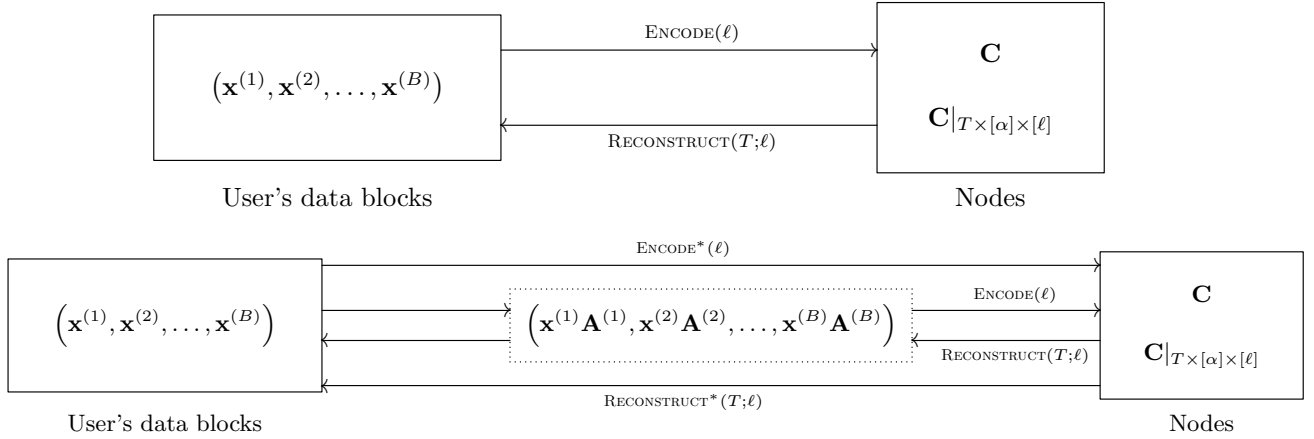

\subsection{Intermediary Encoding}\label{sec:enc1}

Let $\left(\vx^{(1)}, \vx^{(2)}, \ldots, \vx^{(B)}\right)\in\FF_q^{B\times \ell}$ and let 
$\vA^{(1)},\vA^{(2)},\ldots,\vA^{(B)}$ be invertible $\ell\times\ell$ matrices over $\FF_q$.
We define $\enc^*:\FF_q^{ B\times \ell}\to \FF_q^{n\times\alpha\times \ell}$ as

\begin{equation}\label{eq:enc1}
\enc^*(\ell)\left(\vx^{(1)}, \vx^{(2)}, \ldots, \vx^{(B)}\right)\triangleq
\enc(\ell)\left(\vx^{(1)}\vA^{(1)}, \vx^{(2)}\vA^{(2)}, \ldots,\vx^{(B)} \vA^{(B)}\right). \tag{1$*$}
\end{equation}


Next, we show that this encoding function preserves reconstruction and repair capabilities.
The proof for the repair capability is straightforward. For $t\in[n]$ and $T\in \binom{[n]\setminus\{t\}}{d}$, simply 
define $\rep^*(t,T;\ell):\FF_q^{d\times\alpha\times \ell} \to\FF_q^{\alpha \times\ell}$ via
\begin{equation}\label{eq:rep1}
\rep^*(t,T;\ell)(\vC|_{T\times [\alpha]\times [\ell]})\triangleq\rep(t,T;\ell)(\vC|_{T\times [\alpha]\times [\ell]})\tag{3$*$}\\
\end{equation}
Since the right hand side equals $\vC|_{\{t\}\times [\alpha]\times [\ell]}$, the repair property holds.

For the reconstruction property, consider the following algorithm for any $T\in \binom{[n]}{k}$.

\renewcommand*{\algorithmcfname}{Algorithm 2$*$}
\renewcommand{\thealgocf}{}
\begin{algorithm}[H]
\SetAlgoLined
\BlankLine
$\left(\vy^{(1)},\vy^{(2)},\ldots,\vy^{(B)}\right)\gets \rec(T;\ell)(\vC|_{T\times [\alpha]\times [\ell]})$\\
\For{$s\in [B]$}{
$\vhx^{(s)}\gets\vy^{(s)} \left(\vA^{(s)}\right)^{-1}$ \\
}
\Return{$\left(\vhx^{(1)},\vhx^{(2)},\ldots,\vhx^{(B)}\right)$}
\caption{$\rec^*(T;\ell)(\vC|_{T\times [\alpha]\times [\ell]})$\hspace{70mm} (2$*$)}
\end{algorithm}

The reconstruction algorithm produces the correct output since from $\vy^{(s)}=\vx^{(s)}\vA^{(s)}$ for $s\in [B]$, we have $\vhx^{(s)}=\vx^{(s)}$.

We summarize the observations in the following proposition.

\begin{proposition}\label{prop:enc1}
Consider an $(n,k,d,\alpha,B)$ DSS code with functions $\enc$, $\rec$ and $\rep$.
For any positive integer $\ell$ and $B$ invertible $\ell\times\ell$ matrices $\vA^{(1)},\vA^{(2)},\ldots,\vA^{(B)}$,
the functions $\enc^*$, $\rec^*$ and $\rep^*$ given by  \eqref{eq:enc1}, $(2*)$ and \eqref{eq:rep1} describe an $(n,k,d,\alpha\ell,B\ell)$ DSS code.
\end{proposition}

Simple modifications of~\eqref{eq:enc1} that allow for (a) systematic DSS coding and for (b) variable data block lengths are described in what follows.
\vskip 5pt

\begin{enumerate}[(a)]
\item {\noindent \bf Systematic DSS}.
A {\em systematic} DSS code is a DSS code with the property that the $B$ data blocks are explicitly stored amongst a set of $k$ nodes, termed the systematic nodes.
In other words, the content of the systematic nodes are the blocks $\vx^{(s)}$.

Assume a systematic DSS code with functionalities $\enc,\rec,\rep$.
As before, let $\vA^{(1)},\vA^{(2)},\ldots,\vA^{(B)}$ be invertible $\ell\times\ell$ matrices
and define a systematic encoding function $\enc^*(\ell)\left(\vx^{(1)}, \vx^{(2)}, \ldots, \vx^{(B)}\right)$ according to
\begin{align*}
\footnotesize
&\enc^*(\ell)\left.\left(\vx^{(1)}, \vx^{(2)}, \ldots, \vx^{(B)}\right)\right|_{\{t\}\times [\alpha]\times [\ell]}\\ 
&~~~~~~\triangleq\begin{cases}
\enc(\ell)\left.\left(\vx^{(1)}, \vx^{(2)}, \ldots,\vx^{(B)} \right)\right|_{\{t\}\times [\alpha]\times [\ell]}, & \mbox{if $t$ is a systematic node,}\\
\enc(\ell)\left.\left(\vx^{(1)}\vA^{(1)}, \vx^{(2)}\vA^{(2)}, \ldots,\vx^{(B)} \vA^{(B)}\right)\right|_{\{t\}\times [\alpha]\times [\ell]}, & \mbox{otherwise.}
\end{cases}
\end{align*}

Let $\vC$ be the resulting information stored over the $n$ nodes.
The reconstruction and repair algorithms may be used unaltered provided that simple pre-processing of $\vC|_{T\times [\alpha]\times [\ell]}$ is performed first.
Specifically, if $t\in T$ and $t$ is a systematic node, we first let $\vC|_{\{t\}\times[\alpha]\times [\ell]}=\left(\vx^{(s_1)}, \vx^{(s_2)}, \ldots,\vx^{(s_\alpha)} \right)$.
Then, we modify this information to 
\[ \left(\vx^{(s_1)}\vA^{(s_1)}, \vx^{(s_2)}\vA^{(s_2)}, \ldots,\vx^{(s_\alpha)} \vA^{(s_\alpha)}\right). \]
If $\vC'$ denotes the tensor resulting from the previous computation, reconstruction and repair are performed according to $\rec^*(T;\ell)(\vC')$ and $\rep^*(T;\ell)(\vC')$, respectively.
\vskip 5pt

\item {\noindent \bf Data blocks of variable lengths}.
For $s\in[B]$, let $\vx^{(s)}$ be of length $\ell_s$ and let $\vA^{(s)}$ be a {\em right} invertible $\ell_s\times\ell$ matrix,
where $\ell=\max_{s\in [B]} \ell_s$.
Then we define $\enc^*(\ell):\FF_q^{\ell_1}\times\FF_q^{\ell_2}\times\cdots\times\FF_q^{\ell_B}\to\FF_q^{n\times\alpha \times \ell}$ via \eqref{eq:enc1}. If we further define the reconstruction and repair functions via $(2*)$ and \eqref{eq:rep1},
we arrive at an $(n,k,d,\alpha\ell,B\ell)$ DSS code as before.
The right inverse of $\vA^{(s)}$ is required for reconstruction and for simplicity, we use the notation $\left(\vA^{(s)}\right)^{-1}$ for both the inverse and right inverse of a square and rectangular matrix, respectively.
\end{enumerate}

Note that when $\vA^{(s)}=\vI$ for all $s\in [B]$, we recover the traditional encoding function.
In the next subsection, we describe how to choose the matrices $\vA^{(s)}$ and accompanying update protocols so as to ensure significantly lower communication cost between a user and connected storage node as compared to the traditional scheme.

\subsection{Synchronization Schemes}

Suppose the data blocks are edited and the resulting data blocks are of length $\ell'$. 
The idea behind our approach is to request the users to 
modify their respective matrices $\vA^{(s)}$ to be invertible $\ell'\times \ell'$ matrices $\vtA^{(s)}$ according to the edits made. Then, users can only transmit the {\em locations} and {\em values} of their edits rather than a whole span of values, with the storage nodes still being
able to update their respective information so that \eqref{eq:enc1} holds.
Since the matrices $\vtA^{(s)}$ are designed to be invertible,
the resulting system remains an $(n,k,d,\alpha\ell',B\ell')$ DSS code.

We propose three different update schemes for the matrices $\vtA^{(s)}$ based on the frequency and extent to which edits are made, catering to the need of~\cite{weil2006ceph}:
\begin{enumerate}[(4.1)]
\item {\bf Semi-static Data}.
Here, we assume that only a {\em constant fraction} of the data blocks is edited by users so that most data blocks retained their original length $\ell$.
In this case, the matrices $\vtA^{(s)}$ -- albeit modified -- remain of dimension $\ell\times \ell$. The most appropriate choice for the matrices are permutation matrices, i.e., $0$-$1$ matrices with exactly one $1$ per row and per column.
We observe that each node stores $\alpha\ell$ symbols.
 
\item {\bf Hot Data}. 
In contrast to the semi-static case, one may also assume that a significant proportion of the data blocks are edited by users, as is the case for hot data\footnote{In many application, hot data is left uncoded in order to facilitate quick access to information and eliminate the need for re-encoding. Our scheme mitigates both the issues of access and re-encoding, while allowing higher level of data integrity through distributed coding.}. 
In this case, suppose that the resulting data blocks are of length $\ell'<\ell$; then the storage
requirements reduce to each node storing $\alpha\ell'<\alpha\ell$ symbols.
The matrices $\vtA^{(s)}$ have dimension $\ell'\times \ell'$ and an appropriate choice for them are Vandermonde and Cauchy matrices (see~\cite{macwilliams1977theory}), to be discussed in more detail throughout the next sections.
\end{enumerate}

{\noindent \bf Scheme Based on Permutation Matrices}.
 Although this scheme applies for general nonuniform edit models,
and consequently to data blocks of variables lengths, for simplicity we assume edits of the form of a single deletion or insertion. As before, in the nonuniform setting we pad the shorter data blocks by an appropriate number of zeros. 
To do so, we represent the file by padding the data string with zeroes. 
So, a deletion in $\vx$ at position 2 results in $\vtx=(x_1,x_3,\ldots, x_{\ell-1},0)$.

Let data block $\vx^{(s)}$ be edited at coordinate $i_s$.
Recall that we associate with $\vx^{(s)}$ an $\ell\times\ell$ matrix $\vA^{(s)}$.
The matrix $\vA^{(s)}$ is initialized to the identity matrix $\vI$ and 
it remains a permutation matrix after each update. 
Roughly speaking, the storage nodes maintain the coded information in the original order.
Since with each edit this order changes, the permutation matrix $\vA^{(s)}$ is used 
to keep track of the order in the data blocks relative to that in the storage nodes.
Hence, $\vA^{(s)}$ indicates that instead of editing ``position $i_s$'' of the check nodes, one has to edit a position ``position $j_s$'' in the original order. These assertions are stated and proved rigorously in the proof of Proposition \ref{prop:schemep}.

Since permutation matrices are clearly invertible, the approach results in an $(n,k,d,\alpha\ell,B\ell)$ DSS code by Proposition \ref{prop:enc1}.
The permutation matrix intermediary encoding scheme is detailed in Scheme P. 

\renewcommand*{\algorithmcfname}{Scheme P} \label{alg:P}
\begin{algorithm}[H]
\SetAlgoLined
\BlankLine
\eIf{{\rm edit is a deletion}}{
   $j_s\gets$ coordinate  where the $i_s$th row $\vA^{(s)}$ is one (note: $\vA^{(s)}$ is a permutation matrix)\\
   shift $i_s$th row of $\vA^{(s)}$ to the last row
   }{
   $j_s\gets$ coordinate  where the last row $\vA^{(s)}$ is one\\
   shift last row of $\vA^{(s)}$ to the $i_s$th row
  }
User $s$ sends to the connected storage nodes $N(s)$: the value affected, using $x$ ($\log q$ bits), 
the type of edit -- insertion or deletion (one bit), and 
the coordinate $j_s$ ($\log \ell$ bits)\\
\For{$t\in N(s)$}{
	Compute $\vd=\enc(\ell)\left(\vzero,\ldots,\vzero, x\ve_{j_s},\vzero, \ldots,\vzero \right)|_{\{t\}\times [\alpha]\times [\ell]}$\\
	\eIf{{\rm edit is a deletion}}{
	   subtract $\vd$ from coordinate $j_s$ at each storage node
	   }{
	   add $\vd$ to coordinate $j_s$ at each storage node
	  }
 }
\caption{Suppose $\vx^{(s)}$ is edited at coordinate $i_s$.}
\end{algorithm}

\begin{proposition}[Scheme P]\label{prop:schemep}
Consider an $(n,k,d,\alpha\ell,B\ell)$ DSS code and assume a single edit for a single user. 
The updates in accordance to Scheme P result in an $(n,k,d,\alpha\ell, B\ell)$ DSS code and the user needs to communicate $\log\ell+\log q$ bits to a connected storage node to update her/his information.
\end{proposition}

Before we prove this proposition, we illustrate Scheme P via an example.

\begin{example} Consider a simple example of two data blocks $\vu$ and $\vv$, shown below, for which $\ell=5$, as part of a $[3,2]$ MDS code over $\FF_5$. Let $\vD$ denote the tensor given by 
$\enc(\ell)\left(\vzero,\ldots,\vzero, x\ve_{j_s},\vzero, \ldots,\vzero \right)$ in line 10.

\vskip 5pt

{\scriptsize
\centering
\renewcommand{\arraystretch}{1}
\begin{tabular}{ccp{15mm} cccc}\hline

$\vu$ & 
$\vv$ &
edit &
$\vD$ &
$\enc(5)\left(\vu\vA^{(1)},\vv\vA^{(2)}\right)$ &
$\vA^{(1)}$ &
 $\vA^{(2)}$ \\ \hline

$(1,2,3,4,4)$ &
 $(1,1,1,1,1)$ &
--- &
--- &
$\left(
\arraycolsep=2pt
\begin{array}{ccccc}
1 & 2 & 3 & 4 & 4 \\
1 & 1 & 1& 1& 1 \\
2 & 3 & 4 & 0 &0
\end{array}
 \right)$ &
$\left( \arraycolsep=2pt
\begin{array}{ccccc}
1 & 0 & 0 & 0 & 0\\
0 & 1 & 0 & 0 & 0\\
0 & 0 & 1 & 0 & 0\\
0 & 0 & 0 & 1 & 0\\
0 & 0 & 0 & 0 & 1
\end{array}
 \right)$ &
 $\left( \arraycolsep=2pt
\begin{array}{ccccc}
1 & 0 & 0 & 0 & 0\\
0 & 1 & 0 & 0 & 0\\
0 & 0 & 1 & 0 & 0\\
0 & 0 & 0 & 1 & 0\\
0 & 0 & 0 & 0 & 1
\end{array}
 \right)$ \\
 
 $(1,3,4,4,0)$ &
 $(1,1,1,1,1)$ &
deletion at position $2$ of $\vx$ &
$\left(
\arraycolsep=2pt
\begin{array}{ccccc}
0 & 2 & 0 & 0 & 0 \\
0 & 0 & 0 & 0& 0\\
0 & 2 & 0 & 0 &0
\end{array}
 \right)$&
$\left(
\arraycolsep=2pt
\begin{array}{ccccc}
1 & 0 & 3 & 4 & 4 \\
1 & 1 & 1& 1& 1 \\
2 & 1 & 4 & 0 &0
\end{array}
 \right)$ &
$\left( \arraycolsep=2pt
\begin{array}{ccccc}
1 & 0 & 0 & 0 & 0\\
0 & 0 & 1 & 0 & 0\\
0 & 0 & 0 & 1 & 0\\
0 & 0 & 0 & 0 & 1\\
0 & 1 & 0 & 0 & 0
\end{array}
 \right)$ &
 $\left( \arraycolsep=2pt
\begin{array}{ccccc}
1 & 0 & 0 & 0 & 0\\
0 & 1 & 0 & 0 & 0\\
0 & 0 & 1 & 0 & 0\\
0 & 0 & 0 & 1 & 0\\
0 & 0 & 0 & 0 & 1
\end{array}
 \right)$ \\

$(1,3,4,0,0)$ &
 $(1,1,1,1,1)$ &
deletion at position $3$ of $\vx$ &
$\left(
\arraycolsep=2pt
\begin{array}{ccccc}
0 & 0 & 0 & 4 & 0 \\
0 & 0 & 0 & 0& 0\\
0 & 0 & 0 & 4 &0
\end{array}
 \right)$&
$\left(
\arraycolsep=2pt
\begin{array}{ccccc}
1 & 0 & 3 & 0 & 4 \\
1 & 1 & 1& 1& 1 \\
2 & 1 & 4 &1 &0
\end{array}
 \right)$ &
$\left( \arraycolsep=2pt
\begin{array}{ccccc}
1 & 0 & 0 & 0 & 0\\
0 & 0 & 1 & 0 & 0\\
0 & 0 & 0 & 0 & 1\\
0 & 1 & 0 & 0 & 0\\
0 & 0 & 0 & 1 & 0
\end{array}
 \right)$ &
 $\left( \arraycolsep=2pt
\begin{array}{ccccc}
1 & 0 & 0 & 0 & 0\\
0 & 1 & 0 & 0 & 0\\
0 & 0 & 1 & 0 & 0\\
0 & 0 & 0 & 1 & 0\\
0 & 0 & 0 & 0 & 1
\end{array}
 \right)$ \\ 

$(1,4,3,4,0)$ &
 $(1,1,1,1,1)$ &
insertion of $4$ at position $2$ of $\vx$ &
$\left(
\arraycolsep=2pt
\begin{array}{ccccc}
0 & 0 & 0 & 4 & 0 \\
0 & 0 & 0 & 0& 0\\
0 & 0 & 0 & 4 &0
\end{array}
 \right)$&
$\left(
\arraycolsep=2pt
\begin{array}{ccccc}
1 & 0 & 3 & 4 & 4 \\
1 & 1 & 1& 1& 1 \\
2 & 1 & 4 &0 &0
\end{array}
 \right)$ &
$\left( \arraycolsep=2pt
\begin{array}{ccccc}
1 & 0 & 0 & 0 & 0\\
0 & 0 & 0 & 1 & 0\\
0 & 0 & 1 & 0 & 0\\
0 & 0 & 0 & 0 & 1\\
0 & 1 & 0 & 0 & 0\\
\end{array}
 \right)$ &
 $\left( \arraycolsep=2pt
\begin{array}{ccccc}
1 & 0 & 0 & 0 & 0\\
0 & 1 & 0 & 0 & 0\\
0 & 0 & 1 & 0 & 0\\
0 & 0 & 0 & 1 & 0\\
0 & 0 & 0 & 0 & 1
\end{array}
 \right)$ \\ 

 \hline
 
\end{tabular}
}

We reconstruct the data blocks using nodes or rows $1$ and $3$.
\begin{enumerate}
\item From row $1$ and $3$, we infer that row $2$ is $(1,1,1,1,1)$.
\item Hence, we obtain%
\begin{align*}
\vhu&=(1,0,3,4,4)\left(\vA^{(1)}\right)^{-1}= (1,4,3,4,0),\\
\vhv&=(1,1,1,1,1)\left(\vA^{(2)}\right)^{-1}=(1,1,1,1,1),
\end{align*}
\noindent as desired.
\end{enumerate}
\end{example}

\begin{proof}[Proof of Proposition \ref{prop:schemep}]
We demonstrate the correctness of Scheme P for the case of a deletion. Insertions
are handled similarly.

We first show that the updates performed as part of the scheme result in an $(n,k,d,\alpha\ell, B\ell)$ DSS code. It suffices to show that the nodes store 
$\enc(\ell)\left(\vx^{(1)}\vA^{(1)},\vx^{(2)}\vA^{(2)},\ldots,\vtx^{(s)}\vtA^{(s)},\ldots,\vx^{(B)}\vA^{(B)}\right)$, where $\vtA^{(s)}$ is the matrix resulting from instruction performed from 
line 1 to line 7. 
Observe that prior to the edit, the nodes stored   
$\enc(\ell)\left(\vx^{(1)}\vA^{(1)},\vx^{(2)}\vA^{(2)},\ldots,\vx^{(s)}\vA^{(s)},\ldots,\vx^{(B)}\vA^{(B)}\right)$. Define 
\begin{align*}
\vD 
&\triangleq\enc(\ell)\left(\vx^{(1)}\vA^{(1)},\vx^{(2)}\vA^{(2)},\ldots,\vx^{(s)}\vA^{(s)},\ldots,\vx^{(B)}\vA^{(B)}\right)\\
&-\enc(\ell)\left(\vx^{(1)}\vA^{(1)},\vx^{(2)}\vA^{(2)},\ldots,\vtx^{(s)}\vtA^{(s)},\ldots,\vx^{(B)}\vA^{(B)}\right)\\
&=\enc(\ell)\left(\vzero,\ldots,\vzero, \vx^{(s)}\vA^{(s)}-\vtx^{(s)}\vtA^{(s)},\vzero, \ldots,\vzero \right), 
\end{align*}
and suppose that the following claim were true
\be
\vx^{(s)}\vA^{(s)}-\vtx^{(s)}\vtA^{(s)}=x\ve_{j_s}. \label{eq:schemep}
\ee
Then, the updates in lines 9 to 16 have a net effect of subtracting $\vD$ from the information stored at the nodes (recall that $\enc(\ve_s)|_{\{t\}\times[\alpha]}= \vzero$ for $t\notin N(s)$).
This completes the proof.

Hence, it remains to show that \eqref{eq:schemep} is correct.
Without loss of generality, let $\vx^{(s)}=(x_1,x_2,\ldots,x_\ell)$. 
Then $\vtx^{(s)}=(x_1,x_2,\ldots,x_{i_s-1},x_{i_s+1},\ldots,x_\ell,0)$ and $x=x_{i_s}$.

Since $\vA^{(s)}|_{\{i\}\times [\ell]}=\vtA^{(s)}|_{\{i\}\times [\ell]}$ for $i\le i_s-1$,
we have $x_i\vA^{(s)}|_{\{i\}\times [\ell]}-x_i\vtA^{(s)}|_{\{i\}\times [\ell]}=\vzero$.  
Similarly, $x_i\vA^{(s)}|_{\{i\}\times [\ell]}-x_i\vtA^{(s)}|_{\{i-1\}\times [\ell]}=\vzero$ for $i_s+1\le i\le \ell$. Hence, the left hand side yields $x_{i_s}\vtA^{(s)}|_{\{i_s\}\times [\ell]}-0\cdot\vtA^{(s)}|_{\{\ell\}\times [\ell]}=x\ve_{j_s}$, as needed.
\end{proof}


{\noindent \bf Scheme Based on Vandermonde Matrices}.
As for the case of coding with permutation matrices, the scheme using Vandermonde 
matrices applies to nonuniform edits, but for simplicity we assume that there is a single deletion 
at coordinate $i_s$ in data block $\vx^{(s)}$, for $s\in [B]$. Nevertheless, at the end of the section, we briefly outline minor algorithmic 
changes that need to be performed in order to accommodate the non-uniform deletion model.

When one deletion is present, the updated data block length is $\ell'=\ell-1$. Recall that an $\ell \times \ell$ Vandermonde matrix is an 
invertible matrix of the form
\[ \left[ \begin{array}{ccccc}
a_1 & a_1^2 & \ldots & a_1^{\ell-1} & a_1^{\ell} \\
a_2 & a_2^2 & \ldots & a_2^{\ell-1} & a_2^{\ell} \\
\ldots & \ldots & \ldots & \ldots & \ldots \\
a_{\ell-1} & a_{\ell-1}^2 & \ldots & a_{\ell-1}^{\ell-1} & a_{\ell-1}^{\ell} \\
a_{\ell} & a_{\ell}^2 & \ldots & a_{\ell}^{\ell-1} & a_{\ell}^{\ell} 
\end{array}
\right],\]
where $a_1,\ldots,a_{\ell}$ are distinct values over an appropriate field. 

Recall that we associate with $\vx^{(s)}$ an $\ell\times\ell$ matrix $\vA^{(s)}$.
After synchronization, we want the updated matrix $\vtA^{(s)}$ to be of dimension $(\ell-1)\times(\ell-1)$ and invertible, and the information in the $n$ storage nodes reduced to $\alpha\times (\ell-1)$ arrays.

The deleted values $x^{(1)}_{i_1},x^{(2)}_{i_2},\ldots,x^{(B)}_{i_B}$
are stored in the storage nodes as 
\begin{equation}\label{eq:D}
\vD \triangleq \enc(\ell)\left(x^{(1)}_{i_1}\vA|_{\{i_1\}\times[\ell]},x^{(2)}_{i_2}\vA|_{\{i_2\}\times[\ell]},\ldots,
x^{(B)}_{i_B}\vA|_{\{i_B\}\times[\ell]}\right).
\end{equation}
Hence, when given the values $x^{(s)}_{i_s}$ and positions $i_s$,
each node $t$ may subtract the vector $\vD|_{\{t\}\times[\alpha]\times [\ell]}$ from its content. 
To reduce the size of the storage node arrays, we simply remove the coordinates
in the set $[n]\times [\alpha]\times \{\ell\}$. 
Suppose $\vtA^{(s)}$ is the $(\ell-1)\times(\ell-1)$ matrix obtained from $\vA^{(s)}$ by 
removing the $i_s$th row and last column.
It is easy to check that {\em a posteriori} the edit, the storage nodes contain the tensor
\begin{equation}
\enc(\ell-1)\left(\vtx^{(1)}\vtA^{(1)}, \vtx^{(2)}\vtA^{(2)}, \ldots,\vtx^{(B)} \vtA^{(B)}\right).
\label{eq:schemev}
\end{equation}

For the system to be an $(n,k,d,\alpha(\ell-1),B(\ell-1))$ DSS code,
we require $\vtA^{(s)}$ to remain invertible.
This is clearly true if the matrix $\vA^{(s)}$ is Vandermonde.
We refer to the method as Scheme V, the details of which are given in what follows.

\renewcommand*{\algorithmcfname}{Scheme V}
\begin{algorithm}[H]
\SetAlgoLined
\BlankLine
\For{$s\in [B]$}{
User $s$ sends to all connected storage nodes its deleted value deleted -- $x^{(s)}_{i_s}$ -- ($\log q$ bits) 
 as well as the coordinate $i_s$ of the deletion ($\log \ell$ bits)\\
 $\vA^{(s)}\gets  \vA^{(s)}$ via removal of the $i_s$th row and last column
}
\For{$t\in [n]$}{
Using \eqref{eq:D}, compute and subtract $\vD|_{\{t\}\times[\alpha]\times[\ell]}$\\ 
Remove the $(t,j,\ell)$th coordinate, for all $j\in[\alpha]$
}
\caption{{Symbol $\vx^{(s)}$ is deleted at coordinate $i_s$, $s\in [B]$.}}
\end{algorithm}

\begin{proposition}[Scheme V]\label{prop:schemev}
Consider an $(n,k,d,\alpha\ell,B\ell)$ DSS code and assume single deletions in each user data block. 
The updates in accordance to Scheme V result in an $(n,k,d,\alpha(\ell-1), B(\ell-1))$ DSS code 
and each user needs to communicate $\log\ell+\log q$ bits to connected storage nodes to update their information.
\end{proposition}

We again illustrate the communication scheme via an example.

\begin{example} Assume that $\ell=4$ and that the DSS code is a $[3,2]$ MDS code over $\FF_5$. 
As before, we choose two data blocks $\vu$ and $\vv$ as shown below to illustrate the scheme.

\vskip 5pt

{\scriptsize
\begin{center}
\begin{tabular}{cccc cccc}\hline
$\vu$ & 
$i_1$ &
$\vv$ &
$i_2$ &
$\vD$ &
$\enc(\ell)\left(\vu\vA^{(1)},\vv\vA^{(2)}\right)$ &
$\vA^{(1)}$ &
 $\vA^{(2)}$ \\ \hline

 $(0,1,0,1)$ &
--- &
 $(1,0,1,0)$ &
--- &
--- &
$\left(\arraycolsep=2pt
\begin{array}{cccc}
2 & 1 & 0 & 2 \\
2 & 4 & 0 & 3 \\
4 & 0 & 0 & 0
\end{array}
 \right)$ &
$\left(\arraycolsep=2pt
\begin{array}{cccc}
1 & 1 & 1 & 1 \\
1 & 2 & 4 & 3 \\
1 & 3 & 4 & 2 \\
1 & 4 & 1 & 4
\end{array}
 \right)$ &
$\left(\arraycolsep=2pt
\begin{array}{cccc}
1 & 1 & 1 & 1 \\
1 & 2 & 4 & 3 \\
1 & 3 & 4 & 2 \\
1 & 4 & 1 & 4
\end{array}
 \right)$ \\

$(0,1,0)$ &
4 &
 $(0,1,0)$ &
1 &
$\left(\arraycolsep=2pt
\begin{array}{cccc}
1 & 4 & 1 & 4 \\
1 & 1 & 1 & 1 \\
2 & 0 & 2 & 0
\end{array}
 \right)$
&
$\left(\arraycolsep=2pt
\begin{array}{cccc}
1 & 2& 4 & \cancel{3} \\
1 & 3 & 4 & \cancel{2} \\
2 & 0 & 3 & \cancel{0}
\end{array}
 \right)$ &
$\left(\arraycolsep=2pt
\begin{array}{cccc}
1 & 1 & 1  \\
1 & 2 & 4  \\
1 & 3 & 4  \\
\end{array}
 \right)$ &
$\left(\arraycolsep=2pt
\begin{array}{cccc}
1 & 2 & 4   \\
1 & 3 & 4  \\
1 & 4 & 1 
\end{array}
 \right)$ \\
 \hline
 
\end{tabular}
\end{center}
}

We reconstruct the data blocks using nodes or rows $1$ and $3$.
\begin{enumerate}
\item From row $1$ and $3$, we infer that row $2$ is $(1,3,4)$.
\item Hence, we obtain%
\begin{align*}
\vhu&=(1,2,4)\left(\vA^{(1)}\right)^{-1}= (0,1,0),\\
\vhv&=(1,3,4)\left(\vA^{(2)}\right)^{-1}=(0,1,0),
\end{align*}
\noindent as desired.
\end{enumerate}
\end{example}

\begin{proof}[Proof of Proposition \ref{prop:schemev}]
As before, it suffices to show that the information stored at the nodes is given by \eqref{eq:schemev}.
Proceeding in a similar fashion as was done for the proof of Proposition \ref{prop:schemep}, and 
by referring to the linearity of the encoding maps, the proof reduces to showing that for $s\in[B]$,
\be
\left.\left(\vx^{(s)}\vA^{(s)}\right)\right|_{[\ell-1]}
-\vtx^{(s)}\vtA^{(s)}= \left. x^{(s)}_{i_s}\vA \right |_{\{i_s\}\times [\ell-1]}. \label{eq:schemev2}
\ee

As before,
we check that $\left.x^{(s)}_i\vA^{(s)}\right|_{\{i\}\times [\ell-1]}-\left.x^{(s)}_i\vtA^{(s)}\right|_{\{i\}\times [\ell-1]}=\vzero$ for $i\le i_s-1$ and
$\left.x^{(s)}_i\vA^{(s)} \right |_{\{i\}\times [\ell-1]}-\left.x^{(s)}_i\vtA^{(s)} \right |_{\{i-1\}\times [\ell-1]}=\vzero$ for $i_s+1\le i\le \ell$. 
Hence, the remaining term on the left hand side is $\left.x^{(s)}_{i_s}\vA\right|_{\{i_s\}\times [\ell-1]}$, which establishes \eqref{eq:schemev2}.
\end{proof}

\begin{remark}\hfill
\begin{enumerate}[(i)]

\item {\bf Size of $\ell$}. Recall that Vandermonde matrices exist whenever $\ell\le q$, hence this inequality has to hold true for
the scheme to be implementable.

\item {\bf Cauchy matrices based encoding as a means of reducing communication complexity}. In Scheme V, all users were required to 
transmit $(\log\ell+\log q)$ bits to describe their edits. We may save $\log q$ bits in communication complexity by having
one of the users, say user 1, transmit only its location.

We achieve this by fixing $\vA^{(1)}=\vI$. 
Suppose that $\vx^{(1)}$ had a deletion at position $i_1$.
To ensure that $\vtA^{(1)}$ remains invertible, 
we remove the $i_1$th row and $i_1$th column (in line 3).
This in turn forces us to delete the $i_1$th column in all matrices $\vA^{(s)}$.
Hence, one needs to ensure that all square submatrices of $\vA^{(s)}$ are invertible for $s\ge 2$.
It is known that $\ell \times \ell$ Cauchy matrices, taking the general form
\[ \left[ \begin{array}{ccccc}
1/(a_1-b_1) & 1/(a_1-b_2) & \ldots & 1/(a_{1}-b_{\ell-1}) & 1/(a_{1}-b_{\ell}) \\
1/(a_2-b_1) & 1/(a_2-b_2) & \ldots & 1/(a_{2}-b_{\ell-1}) & 1/(a_{2}-b_{\ell}) \\
\ldots & \ldots & \ldots & \ldots & \ldots \\
1/(a_{\ell-1}-b_1) & 1/(a_{\ell-1}-b_2) & \ldots & 1/(a_{\ell-1}-b_{\ell-1}) & 1/(a_{\ell-1}-b_{\ell}) \\
1/(a_{\ell}-b_1) & 1/(a_{\ell}-b_2) & \ldots & 1/(a_{\ell}-b_{\ell-1}) & 1/(a_{\ell}-b_{\ell}) \\
\end{array}
\right],\]
for distinct values $a_1,\ldots,a_{\ell},b_1,\ldots,b_{\ell},$ satisfy this requirement and clearly 
exist whenever $2\ell\le q$. Hence, for the protocol to be implementable, the alphabet size
and length of data blocks need to satisfy the constraint $2\ell\le q$.

\item {\bf Application to data deduplication.} Deduplication broadly refers to the process of 
removing duplicate copies of data with the objective of saving storage \cite{meister2009multi}.
%
Scheme V may easily be integrated into a data deduplication process for a DSS
so as to remove duplicates not only amongst the users, 
but also their redundantly encoded information at the storage nodes.

We describe how to accomplish this task for {\em post-process} deduplication, i.e., deduplication after the users have already written on their disks certain data blocks, say $(f_1,f_2,\ldots,f_e)\in\FF_q^e$.
Deduplication proceeds as follows:
\begin{enumerate}[(I)]
\item A central node broadcasts to all users and nodes the data $(f_1,f_2,\ldots,f_e)$ to be removed.
\item For $s\in [B]$, user $s$ scans the string $\vx^{(s)}$ for the data string  $(f_1,f_2,\ldots,f_e)$
and identifies positions $i_{s,1}, i_{s,2},\ldots,i_{s,e}$ where the blocks are stored.
\item User $s$ transmits positions $i_{s,1}, i_{s,2},\ldots,i_{s,e}$ to all connected storage nodes.
\item Each storage node and user updates information as requested by Scheme V in $e$ iterations.
\end{enumerate}
\end{enumerate}
\end{remark}

{\noindent \bf Nonuniform Deletion Model}.
To conclude this subsection, we describe how to make simple modifications to Scheme V so as to
adapt it to the general scenario where each user $\vx^{(s)}$ has $d_s$ deletions. 
Define $\ell'$ as $\max_{s\in[B]} \ell-d_s$, or equivalently, $\ell-\min_{s\in[B]} d_s$.
Our goal is to updated the matrix $\vtA^{(s)}$ so that it has dimension $(\ell-d_s)\times\ell'$, 
and reduce the content of the $n$ storage nodes to $\alpha \times \ell'$ arrays.

Suppose $\vx^{(s)}$ has deleted values $x^{(s)}_{i_{s1}},x^{(s)}_{i_{s2}},\ldots,x^{(s)}_{i_{sd_s}}$.
Define $\vd^{(s)}\triangleq \sum_{t=1}^{d_s} x^{(s)}_{i_{st}}\vA|_{\{i_{st}\}\times[\ell]}$ and modify \eqref{eq:D} as 
\begin{equation}\label{eq:D2}
\vD \triangleq \enc(\ell)\left(\vdd^{(1)},\vdd^{(2)},\ldots, \vdd^{(B)}\right). \tag{$8'$}
\end{equation}
As before, when given the deleted values and their coordinates,
each node $t$ subtracts the vector $\vD|_{\{t\}\times[\alpha]\times[\ell]}$. 
To reduce the size of the storage nodes, we simply remove 
the coordinates belonging to $[n] \times [\alpha] \times \{i\}$ for $\ell'+1\le i\le \ell$.
Suppose $\vtA^{(s)}$ is the $(\ell-d_s)\times\ell'$ matrix 
with the rows corresponding to the $d_s$ deletions and last $\ell-\ell'$ columns removed from $\vA^{(s)}$.
Then it is not difficult to check that the matrices $\vtA^{(s)}$ remain Vandermonde if the $\vA^{(s)}$ matrices
were initially chose to be Vandermonde.

Next, we establish some fundamental limits of communication complexity for the given synchronization of DSS
codes under edit protocols, and compare the performance of schemes P, V, and T with these limits.

\section{Fundamental Limits and a Tradeoff between Communication Cost and Storage Overhead}
\label{sec:fund+trade}

Suppose a user's data block is subjected to a {\em single} edit. 
Then the communication cost of both Scheme P and Scheme V is 
approximately $\log \ell+\log q$, for each pair of user and his/her connected storage node.
Hence, for a single edit, the communication cost of both schemes is near-optimal.

However, when a data block has an arbitrary number of edits, say $d$,
then the schemes require $d(\log\ell+\log q)$ bits to be communicated and it is unclear if this 
quantity is optimal, order optimal (i.e., of the same order as the optimal solution) or suboptimal.
In the next Section \ref{sec:fundamental}, we establish a simple lower bound on the communication cost
using results of Levenshtein~\cite{Levenshtein:2001} and show that Scheme P and Scheme V are within a constant factor away from the lower bound
when $d=o(\ell^{1-\epsilon})$ for a constant $0<\epsilon<1$. 
Hence, under these conditions for $d$, the schemes are order optimal.

Nevertheless, it is worth pointing out that Scheme P and Scheme V outperform Scheme T for any number of edits.
To achieve the communication cost of $O(\log \ell+\log q)$ bits,
Scheme P and Scheme V have to store certain structural information regarding the matrices $\vA^{(s)}$. 
As Scheme T does not require this storage overhead, we also analyze the tradeoff between communication cost and storage overhead 
in Section \ref{sec:tradeoff}. Our findings suggest that the use of Scheme P and Scheme V is preferred to the use of scheme T
when the number of edits is $d=o(\ell/\log\ell)$.

\subsection{Fundamental Limits}\label{sec:fundamental}

We provide next a lower bound on the communication cost between a user and a connected storage node and 
show that Scheme P and Scheme V are within a constant factor away from this lower bound
provided the number of edits is constant. 

Recall that a subsequence of a sequence is itself a sequence that can be obtained from the original sequence by deleting some elements without changing the order of the nonedited elements. Similarly, a supersequence of a sequence is itself a sequence that can be obtained from the original sequence by inserting some elements without changing the order of the nonedited elements. 

Consider the following quantity that counts the number of possible subsequences or supersequences 
resulting from $d$ edits on a data block of length $\ell$, respectively, defined via
\begin{equation}
V(\ell,d)\triangleq 
\begin{cases} 
\max_{\vx\in\FF_q^\ell} | \{\vtx: \vtx \mbox{ result from at most $d$ deletions from } \vx\}|,\\
\max_{\vx\in\FF_q^\ell} | \{\vtx: \vtx \mbox{ result from at most $d$ insertions and deletions from } \vx\}|.
\end{cases}
\end{equation}

We require $\log V(\ell,d)$ bits to describe $d$ deletions (or  $d$ insertions/deletions) to a node that contains the original sequence.
This shows that each user needs to communicate to a connected storage node 
at least $\log V(\ell,d)$ bits, given that the storage nodes contain only coded information about the original sequence and that the user knows the positions of the edits. A similar argument for establishing a lower bound for the classical two-way communication protocol was used in~\cite{orlitsky1987communication}. 

The fundamental combinatorial entity $V(\ell,d)$ was introduced by Levenshtein~\cite{Levenshtein:1966} and has since been studied by a number of authors \cite{Levenshtein:1966, Calabi.Hartnett:1969, Hirschberg.Regnier:2000, Levenshtein:2001, Mercier.etal:2008, Liron.Langberg:2012}. In particular, it is known (see~\cite[Eq. (11) and (24)]{Levenshtein:2001}) that
\begin{equation*}
\log V(\ell,d)\ge
\begin{cases} 
\log\binom{\ell-d}{d}, & \mbox{for deletions only},\\
\log (q-1)^d\binom{\ell+d}{d}, & \mbox{for deletions and insertions}.
\end{cases}
\end{equation*}

When $d=o(\ell^{1-\epsilon})$ for a constant $0<\epsilon<1$, 
in the asymptotic parameter regime we have $\log\binom{\ell-d}{d}=\Omega(\log \ell)$ and 
$\log (q-1)^d\binom{\ell+d}{d}=\Omega(\log \ell+\log q)$.
Therefore, Scheme P and Scheme V are within a constant factor away from this lower bound.

\subsection{Tradeoff between Communication Cost and Storage Overhead}
\label{sec:tradeoff}

Suppose that during each round of update, a data block has a single edit.
In this case, the communication cost between each user and each connected storage node for 
both Scheme P and Scheme V is approximately $\log \ell+\log q$.
However, to achieve this cost, each user needs to store 
the associated matrix $\vA^{(\cdot)}$. Hence, each user a priori requires 
$\ell^2 \log q$ bits to store this matrix.
%

However, this stringent storage requirement may be easily relaxed. For any synchronization scheme, 
we can first store the description of the initial matrices%
\footnote{The initial matrix is the identity matrix for Scheme P and Scheme T, 
and a Vandermonde matrix for Scheme V. }.
Subsequently, to generate the matrices $\vA^{(s)}$, it suffices for the users to store the modifications 
to the initial matrices after the edits
and we term this information the {\em storage overhead per edit}.
For Scheme P and Scheme V, this overhead amounts to the space allocated for storing 
the locations of edits; hence, the storage overhead per edit is $\log \ell$ bits.
In contrast, for Scheme T, the storage overhead per edit is zero as the matrices $\vA^{(\cdot)}$ are 
always the identity. 

We summarize the communication cost and storage overhead features of the various schemes in Table~\ref{tab:features}. 
\vskip 5pt
{
\begin{table}
\centering
\begin{tabular}{|p{6cm}|p{4cm}|p{2.5cm}|p{2.5cm}|} 
\hline
& Scheme T & Scheme P & Scheme V\\ \hline\hline
Communication Cost between User and Storage Node per Edit  &
$(\ell-1)\log q$ (worst case) &
$\log \ell+\log q+1$ &
$\log \ell+\log q$ \\ \hline
Storage Overhead per Edit &
0 &
$\log \ell$ &
$\log \ell$ \\ \hline
Applicability &
--- &
---&
$\ell\le q$ and for deletions only \\ \hline
Total Size of Storage Nodes for Data Blocks of Length $\ell'$ &
$n\alpha\ell'$ &
$n\alpha\ell$, where $\ell$ is the length of the original data blocks&
$n\alpha\ell'$ \\ \hline
\end{tabular}
\caption{Trade-off between storage overhead and communication complexity of various synchronization protocols.}
\label{tab:features}

\end{table} 
}
\vskip 5pt

{\noindent \bf Storage overhead versus information storage}. 
Consider a single data block that has $d$ edits. Then both Scheme P and Scheme V incur a total storage overhead of $d\log\ell$ bits for the user.
But the data block itself is of size $\ell\log q$ bits. 
Therefore, for desirable storage allocation properties, one would want $d\log\ell=o(\ell)$ or the number of edits to be $o(\ell/\log\ell)$. This implies that Schemes P and V should be used only in the small/moderate edit regime. 
Nevertheless, it is preferable to use Schemes P and V for semi-static and hot data, respectively, since
the former effectively maintains the length of the original data file, while the latter scales storage requirements according to the number of unedited symbols.
\vskip 5pt

{\noindent \bf Hybrid Schemes}. 
Suppose first that we are given a constraint $\Gamma$ on the storage overhead per edit, say $0<\Gamma<\log\ell$.
By combining Scheme V and Scheme T, 
we demonstrate a scheme that achieves a lower (worst case) communication cost 
while satisfying the aforementioned storage constraint.

Pick $\gamma$ so that $\log(\gamma\ell)\le\Gamma$.
Next, divide each data block into two parts of lengths approximately $\gamma\ell$ and $(1-\gamma)\ell$.
If the edit belongs to the first part, proceed with Scheme V by sending out $\log\ell+\log q$ bits
and storing $\log(\gamma\ell)$ bits.
On the other hand, if the edit belongs to the second part, 
then simply send the entire second part of $(1-\gamma)\ell\log q$ bits and the deleted position of $\log\ell$ bits
with no storage overhead.
Hence, in the worst case, the hybrid scheme incurs a storage overhead of $\max\{\log(\gamma\ell),0\}=\log(\gamma\ell)$ bits and a communication cost of  $\max\{\log(\gamma\ell)+\log q,(1-\gamma)\ell\log q\}$ bits. The choice of $\gamma$
that minimizes the communication cost is given by the next lemma.
\begin{lemma}
The hybrid scheme has smallest communication cost if
\begin{equation}
\gamma=\min\left\{\frac{W(q^{\ell-1}\,\log\,q)}{\ell\, \log\,q}, \frac{2^\Gamma}{\ell}\right\},
\end{equation}
where $W(x)$ denotes the Lambert function, defined via $x=W(x)\,\exp(W(x))$.
\end{lemma}
\begin{proof}
The result easily follows by observing that $\max\{\log(\gamma\ell)+\log q,(1-\gamma)\ell\log q\}$ is minimized
when $\log(\gamma\ell)+\log q=(1-\gamma)\ell\log q$. By rearranging terms, we arrive at the expression $z=y\, \exp{y}$
with $y=\gamma\, \ell \log\,q$ and $z=q^{\ell+1}\, \log\, q$, which yields the Lambert function form. The proof follows by observing that $\gamma$ is also required to satisfy $\log(\gamma\, \ell) \leq \Gamma$.
\end{proof}

The hybrid scheme is described in detail in Appendix \ref{sec:schemeh} and
Figure \ref{fig:tradeoff} illustrates the communication cost and storage overhead trade-offs achievable by the hybrid scheme. 
Observe that both Scheme P and Scheme V come close to the lower bound when one is allowed a storage overhead of $\log\ell$ bits, assuming single edits per data block.
Note that the lower bound on the communication cost is independent on the storage overhead. It would hence be of interest to derive a lower bound that actually captures the dependency on the storage overhead.

\begin{figure}[!t]
\begin{center}
\caption{Communication Cost and Storage Overhead Tradeoff}
\label{fig:tradeoff}
\includegraphics[scale=0.5]{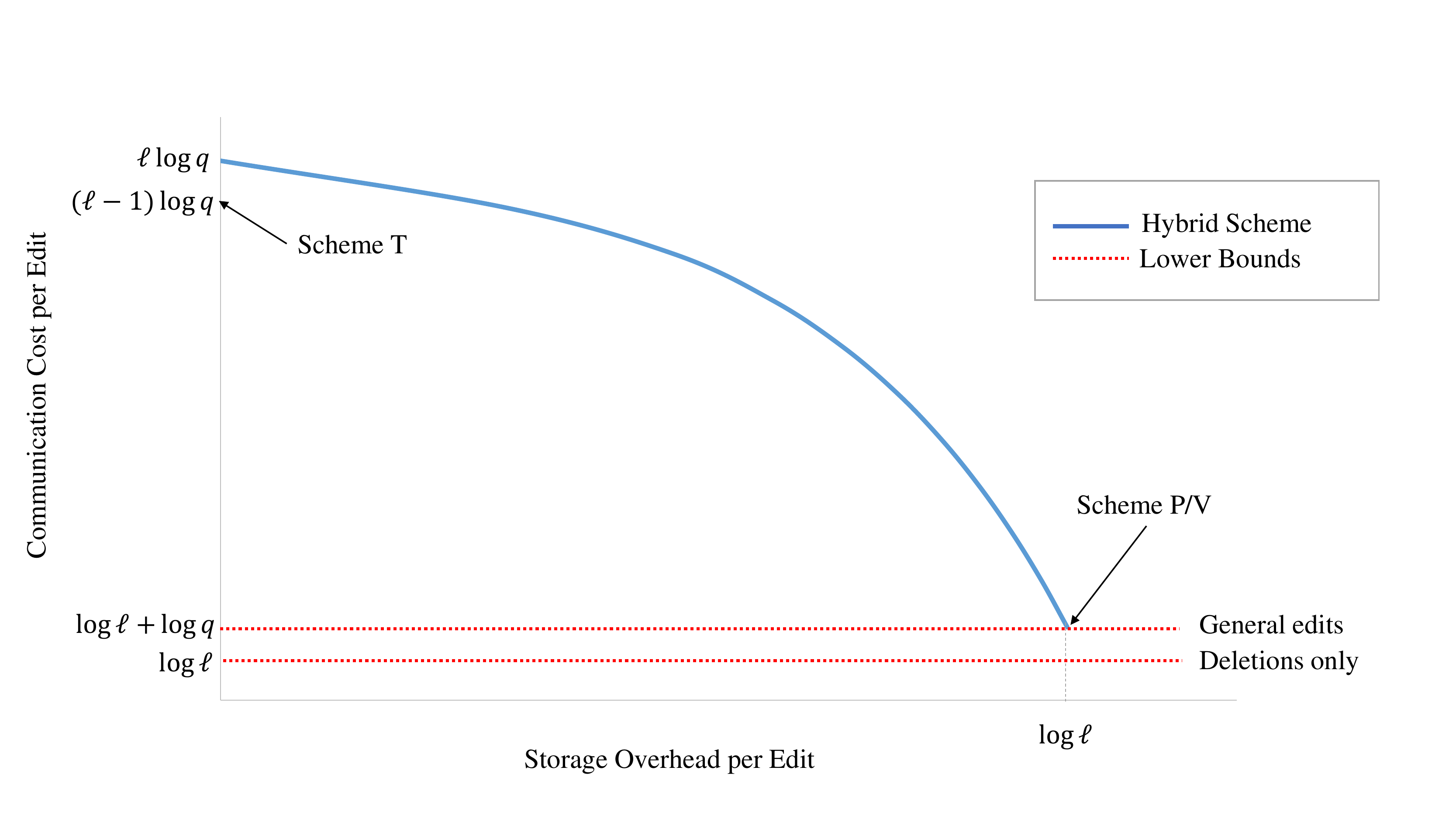}
\end{center}
\end{figure}

We conclude this section by considering a cost that takes into account {\em both} communication complexity and storage overhead. Similar to the hybrid scheme, we request certain users to communicate via Scheme P/V and others via Scheme T. Our new objective is to find an assignment that minimizes the aggregate cost of storage and communication, rather than an assignment that obeys a storage overhead bound only.

Specifically, we consider a probabilistic edit model, where
one edit occurs in each round and
the probability of data block $s$ being edited equals $p_s$.
Let $V$ be the set of users that are assigned Scheme P/V
and define $p\triangleq \sum_{s\in V} p_s$.

Then, the expected storage overhead is $p\log\ell$, 
while the expected communication cost is $\log\ell+p\log q+ (1-p)\ell\log q$.
For some predefined $\theta\ge 0$, the aggregate cost is defined as
\begin{equation}\label{eq:aggregate}
C_A(V)\triangleq \log\ell+p\log q+ (1-p)\ell\log q + \theta p\log\ell.
\end{equation}

\begin{proposition}
Let $C_A$ be given by \eqref{eq:aggregate}.
If $\theta\le (\ell-1)\log q/\log \ell$, then $C_A$ is minimized when $V=[B]$.
Otherwise, if $\theta> (\ell-1)\log q/\log \ell$, then $C_A$ is minimized when $V=\varnothing$.
\end{proposition}

\begin{proof}
Consider $\theta\le (\ell-1)\log q/\log \ell$. Suppose $V\ne [B]$ and pick $s\notin V$.
Let $V'=V\cup\{s\}$. Then 
\begin{equation*}
C_A(V')-C_A(V)=p_s \log q - p_s \ell\log q+\theta p_s\log \ell=p_s(\theta\log\ell -(\ell-1)\log q)\le 0.
\end{equation*}
Hence, augmenting $V$ with $s$ lowers the aggregate cost and so $C_A$ is minimized when $V=[B]$.
The argument for the other case proceeds along similar lines.
\end{proof}


\section{Unknown Deletion Locations}

In the previous sections, we considered synchronization protocols for user induced edits. In this case, the positions of deletions and the values of the corresponding symbols are known and available to the users of the DSS. Many applications, including Dropbox-like file synchronization, call
for file updates when changes are made by secondary users, in which case the edits are unknown to primary or
other file users. As will be shown next, it is straightforward to accommodate the update scenarios to this model, 
provided that an additional small storage overhead is allowed. 
Again, we focus on the single edit scenario, for which
we show that one only needs to store the Varshamov-Tenengolts (VT) Syndrome~\cite{sloane2002single} of the data, in addition to 
a properly encoded original file. The single edit scenario applies to the case when frequent checks or updates are performed. When the updates appear less frequently, one may apply multiple-deletion correcting codes akin to those described in~\cite{davey2001reliable}.
%
%

Consider the data string $\vx^{(s)}$.
Recall that its VT syndrome is given by 
\begin{align*}
\nu_1^{(s)} &\triangleq
\sum_{i=1}^\ell x^{(s)}_i,\\
\nu^{(s)}_2 &\triangleq
\sum_{i=1}^{\ell-1} i\chi_i \bmod{\ell},
\end{align*} 
where $\chi_i$ is the indicator function for the event $x^{(s)}_i\le x^{(s)}_{i+1}$, $i\in[\ell-1]$.

{\noindent \bf Case when all users have a single deletion}.
In addition to the matrices $\vA^{(s)}$, we have each user store its VT syndrome $(\nu_1^{(s)},\nu_2^{(s)})$, 
which is of size $\log q + \log \ell$ bits. Then prior to an update, each user retrieves its VT syndrome to compute the deleted position and its value. With this information, the user proceeds with any of the previously outlined 
update schemes.

{\noindent \bf Case where a proportion of data blocks have a deletion}.
To simplify arguments, we consider a systematic $[n,k]$ MDS code 
and the argument for a systematic $[n,k,d]$ regenerating code can be derived as a straightforward extension of the
former case. 

Observe that in the previous scheme, an additional overhead of $k(\log q+\log \ell)$ bits was required 
to store all VT syndromes for all $k$ users. We next show that by using the structure of the MDS code,
we can achieve storage savings by recording the VT syndromes of the $(n-k)$ check nodes only.
For this purpose, assume that at most $n-k$ data blocks have a deletion and that $n-k<k$.
Also, for simplicity, let $\ell= q$.

For the syndromes $\nu_1^{(1)},\nu_1^{(2)},\ldots, \nu_1^{(k)}$, 
each check node stores an additional check value that represents a linear combination of these syndromes,
so that from any collection of $k$ nodes all syndromes can be recovered.
Similarly, each check node stores another check value from the syndromes%
\footnote{Here, we assume $q$ is a prime and the syndrome is computed modulo $q$.}
$\nu_2^{(1)},\nu_2^{(2)},\ldots, \nu_2^{(k)}$.
Therefore, in this scheme, we store an additional $2(n-k)\log q<k(\log q+\log \ell)$ bits.
Suppose that $n-k$ affected data blocks have a deletion. 
Each affected user requests the VT syndromes from the unaffected $k-(n-k)$ users
and also the coded VT syndromes from the $n-k$ check nodes.
With the $k$ values, the affected user computes its own VT syndrome and consequently uses it 
to determine the position of its deletion and its value.

\section{Conclusion}

We presented a collection of protocols for synchronization of data encoded using regenerating and repair distributed storage codes, under probabilistic and deterministic symbol deletion edits. Our protocols trade communication overhead for small increases in storage overheads, 
and may be applied in various data deduplication and system update scenarios. The gist of the approach introduced was to use intermediary encoding, which allows the updated encoded information to retain repair and reconstruction properties, although with a different encoding functionality. In this case, the synchronization protocols may be seen as a form of functional update rules -- rules that preserve functional properties but not the code structure. For the presented protocols, we provided simple fundamental lower bounds, average case performance bounds as well as an average case complexity analysis.

\bibliographystyle{ieeetr}
\bibliography{DSS}

\appendix

\section{Hybrid Schemes}
\label{sec:schemeh}

We describe in detail the hybrid scheme that combines Scheme P and Scheme V.
Define $\ell^*=(1-\gamma)\ell$. 
Then, for $s\in[B]$, initialize 
\begin{equation*}
\vA^{(s)}\gets 
\left(
\begin{array}{cc}
\vV & \vzero\\
\vzero & \vI
\end{array}
\right),
\end{equation*}
\noindent where $\vV$ is a Vandermonde matrix of dimension $(\ell-\ell^*)\times(\ell-\ell^*)$
and $\vI$ is an identity matrix of dimension $\ell^*$.
For subsequent edits, we proceed as below.

\renewcommand*{\algorithmcfname}{Scheme H}
\begin{algorithm}[H]
\SetAlgoLined
\BlankLine
\eIf{$i_s\le \ell-\ell^*$}{
User $s$ sends to all connected storage nodes $N(s)$ the value of the deleted symbol -- $x^{(s)}_{i_s}$ ($\log q$ bits) 
 and the coordinate $i_s$ ($\log \ell$ bits)\\
 $\vA^{(s)}\gets  \vA^{(s)}$ with the $i_s$th row  removed
}
{
User $s$ computes the vector $\vd\gets \left(x^{(s)}_{\ell-\ell^*+1},x^{(s)}_{\ell-\ell^*+2},\ldots,x^{(s)}_{\ell}\right)-
\left(x^{(s)}_{\ell-\ell^*+1},\ldots,x^{(s)}_{i_s-1},x^{(s)}_{i_s+1},\ldots,x^{(s)}_{\ell},0\right)$\\
User $s$ sends to all connected storage nodes $N(s)$ the vector $\vd$  ($\ell^*\log q$ bits)\\
 User $s$ pads $\vx^{(s)}$ with a zero.
}

\For{$t\in N(s)$}{
Subtract $\vD|_{\{t\}\times[\alpha]\times[\ell]}$ where 
\begin{equation*}
\vD=
\begin{cases}
\enc(\ell)\left(\vzero,\vzero,\ldots,x^{(s)}_{i_s}\vA|_{\{i_s\}\times[\ell]},\ldots,\vzero\right), & \mbox{if $i_s\le \ell-\ell^*$},\\
\enc(\ell)\left(\vzero,\vzero,\ldots,(\vzero,\vd),\ldots,\vzero\right), & \mbox{otherwise}.\\
\end{cases}
\end{equation*}
}
\caption{{Assumption: $\vx^{(s)}$ is deleted at coordinate $i_s$.}}
\end{algorithm}

\vskip 5pt

\begin{proposition}[Scheme H]\label{prop:schemeh}
Consider an $(n,k,d,\alpha\ell,B\ell)$ DSS code and assume a single deletion per user. 
The updates in accordance to Scheme H result in an $(n,k,d,\alpha\ell, B\ell)$ DSS code.
In the worst case, each user introduces a storage overhead of $\max\{\log(\gamma\ell),0\}$ bits  
and incurs a communication cost of  $\max\{\log(\gamma\ell)+\log q,(1-\gamma)\ell\log q\}$ bits for its edit.
\end{proposition}

The proof is similar to the proofs of Proposition \ref{prop:schemep} and Proposition \ref{prop:schemev}
and therefore omitted.

\pagebreak

\begin{example} We illustrate Scheme H via a simple example involving a  $[3,2]$ MDS code over $\FF_5$, involving two arbitrarily chosen data blocks $\vu$ and $\vv$ shown below. 

\vskip 5pt

{\scriptsize
\centering
\renewcommand{\arraystretch}{1}
\begin{tabular}{cc cccc}\hline
$\vu$ & 
$\vv$ &
$\vD$ &
$\enc(\ell)\left(\vu\vA^{(1)},\vv\vA^{(2)}\right)$ &
$\vA^{(1)}$ &
 $\vA^{(2)}$ \\ \hline

$(1,1,1,1,1,1,1)$ &
 $(1,2,3,4,3,2,1)$ &
--- &
$\left(
\arraycolsep=2pt
\begin{array}{ccccccc}
4 & 0 & 0 & 0 & 1 & 1 & 1\\
0 & 0 & 0 & 4 & 3 & 2 & 1\\
4 & 0 & 0 & 4 & 4 & 3 & 2\\
\end{array}
 \right)$ &
$\left( \arraycolsep=2pt
\begin{array}{ccccccc}
1 & 1 & 1 & 1 & 0 & 0 & 0\\
1 & 2 & 4 & 3 & 0 & 0 & 0\\
1 & 3 & 4 & 2 & 0 & 0 & 0\\
1 & 4 & 1 & 4 & 0 & 0 & 0\\
0 & 0 & 0 & 0 & 1 & 0 & 0\\
0 & 0 & 0 & 0 & 0 & 1 & 0\\
0 & 0 & 0 & 0 & 0 & 0 & 1
\end{array}
 \right)$ &
 $\left( \arraycolsep=2pt
\begin{array}{ccccccc}
1 & 1 & 1 & 1 & 0 & 0 & 0\\
1 & 2 & 4 & 3 & 0 & 0 & 0\\
1 & 3 & 4 & 2 & 0 & 0 & 0\\
1 & 4 & 1 & 4 & 0 & 0 & 0\\
0 & 0 & 0 & 0 & 1 & 0 & 0\\
0 & 0 & 0 & 0 & 0 & 1 & 0\\
0 & 0 & 0 & 0 & 0 & 0 & 1
\end{array}
 \right)$ \\

$(1,1,\cancel{1},1,1,1,1)$ &
 $(1,2,3,4,3,2,1)$ &
$\left(
\arraycolsep=2pt
\begin{array}{ccccccc}
1 & 3 & 4 & 2 & 0 & 0 & 0\\
0 & 0 & 0 & 0 & 0 & 0 & 0\\
1 & 3 & 4 & 2 & 0 & 0 & 0\\
\end{array}
 \right)$ &
$\left(
\arraycolsep=2pt
\begin{array}{ccccccc}
3 & 2 & 1 & 3 & 1 & 1 & 1\\
0 & 0 & 0 & 4 & 3 & 2 & 1\\
3 & 2 & 1 & 2 & 4 & 3 & 2\\
\end{array}
 \right)$ &
$\left( \arraycolsep=2pt
\begin{array}{ccccccc}
1 & 1 & 1 & 1 & 0 & 0 & 0\\
1 & 2 & 4 & 3 & 0 & 0 & 0\\
1 & 4 & 1 & 4 & 0 & 0 & 0\\
0 & 0 & 0 & 0 & 1 & 0 & 0\\
0 & 0 & 0 & 0 & 0 & 1 & 0\\
0 & 0 & 0 & 0 & 0 & 0 & 1
\end{array}
 \right)$ &
 $\left( \arraycolsep=2pt
\begin{array}{ccccccc}
1 & 1 & 1 & 1 & 0 & 0 & 0\\
1 & 2 & 4 & 3 & 0 & 0 & 0\\
1 & 3 & 4 & 2 & 0 & 0 & 0\\
1 & 4 & 1 & 4 & 0 & 0 & 0\\
0 & 0 & 0 & 0 & 1 & 0 & 0\\
0 & 0 & 0 & 0 & 0 & 1 & 0\\
0 & 0 & 0 & 0 & 0 & 0 & 1
\end{array}
 \right)$ \\

$(1,1,1,1,1,1)$ &
 $(1,2,3,4,\cancel{3},2,1)$ &
$\left(
\arraycolsep=2pt
\begin{array}{ccccccc}
0 & 0 & 0 & 0 & 0 & 0 & 0\\
0 & 0 & 0 & 0 & 1 & 1 & 1\\
0 & 0 & 0 & 0 & 1 & 1 & 1\\
\end{array}
 \right)$ &
$\left(
\arraycolsep=2pt
\begin{array}{ccccccc}
3 & 2 & 1 & 3 & 1 & 1 & 1\\
0 & 0 & 0 & 4 & 2 & 1 & 0\\
3 & 2 & 1 & 2 & 3 & 2 & 1\\
\end{array}
 \right)$ &
$\left( \arraycolsep=2pt
\begin{array}{ccccccc}
1 & 1 & 1 & 1 & 0 & 0 & 0\\
1 & 2 & 4 & 3 & 0 & 0 & 0\\
1 & 4 & 1 & 4 & 0 & 0 & 0\\
0 & 0 & 0 & 0 & 1 & 0 & 0\\
0 & 0 & 0 & 0 & 0 & 1 & 0\\
0 & 0 & 0 & 0 & 0 & 0 & 1
\end{array}
 \right)$ &
 $\left( \arraycolsep=2pt
\begin{array}{ccccccc}
1 & 1 & 1 & 1 & 0 & 0 & 0\\
1 & 2 & 4 & 3 & 0 & 0 & 0\\
1 & 3 & 4 & 2 & 0 & 0 & 0\\
1 & 4 & 1 & 4 & 0 & 0 & 0\\
0 & 0 & 0 & 0 & 1 & 0 & 0\\
0 & 0 & 0 & 0 & 0 & 1 & 0\\
0 & 0 & 0 & 0 & 0 & 0 & 1
\end{array}
 \right)$ \\

$(1,1,1,1,1,1)$ &
 $(\cancel{1},2,3,4,2,1,0)$ &
$\left(
\arraycolsep=2pt
\begin{array}{ccccccc}
0 & 0 & 0 & 0 & 0 & 0 & 0\\
1 & 1 & 1 & 1 & 0 & 0 & 0\\
1 & 1 & 1 & 1 & 0 & 0 & 0\\
\end{array}
 \right)$ &
$\left(
\arraycolsep=2pt
\begin{array}{ccccccc}
3 & 2 & 1 & 3 & 1 & 1 & 1\\
4 & 4 & 4 & 3 & 2 & 1 & 0\\
2 & 1 & 0 & 1 & 3 & 2 & 1\\
\end{array}
 \right)$ &
$\left( \arraycolsep=2pt
\begin{array}{ccccccc}
1 & 1 & 1 & 1 & 0 & 0 & 0\\
1 & 2 & 4 & 3 & 0 & 0 & 0\\
1 & 4 & 1 & 4 & 0 & 0 & 0\\
0 & 0 & 0 & 0 & 1 & 0 & 0\\
0 & 0 & 0 & 0 & 0 & 1 & 0\\
0 & 0 & 0 & 0 & 0 & 0 & 1
\end{array}
 \right)$ &
 $\left( \arraycolsep=2pt
\begin{array}{ccccccc}
1 & 2 & 4 & 3 & 0 & 0 & 0\\
1 & 3 & 4 & 2 & 0 & 0 & 0\\
1 & 4 & 1 & 4 & 0 & 0 & 0\\
0 & 0 & 0 & 0 & 1 & 0 & 0\\
0 & 0 & 0 & 0 & 0 & 1 & 0\\
0 & 0 & 0 & 0 & 0 & 0 & 1
\end{array}
 \right)$ \\
 \hline
\end{tabular}
}

We reconstruct the data blocks using nodes or rows $1$ and $3$.
\begin{enumerate}
\item From row $1$ and $3$, we infer that row $2$ is $(4,4,4,3,2,1,0)$.
\item Hence, we obtain%
\begin{align*}
\vhu&=(3,2,1,3,1,1,1)\left(\vA^{(1)}\right)^{-1}= (1,1,1,1,1,1),\\
\vhv&=(4,4,4,3,2,1,0)\left(\vA^{(2)}\right)^{-1}=(2,3,4,2,1,0),
\end{align*}
\noindent as desired.
\end{enumerate}

\end{example}

\section{Probabilistic Analysis of the Communication Cost}
\label{sec:prob}

In Section \ref{sec:tradeoff}, we described a straightforward worst case analysis of the communication costs incurred by various synchronization/update protocols. In what follows, we estimate the expected communication costs of the schemes in the asymptotic regime, and in particular, show that even on average, Scheme T requires communicating $\Omega(\ell\log q)$ bits per one single round of editing, unlike schemes P and V.

For the purpose of average case analysis, we assume the following probabilistic models for edits:
\begin{enumerate}
\item[(UD)] The uniform deletions model. For simplicity, the underlying assumption is that each data block has a single deletion, at a coordinate chosen uniformly at random.

\item[(CND)] The combinatorial nonuniform deletion model. Here, the model assumes that exactly $D$ deletions occurred uniformly at random over all the $\ell$ coordinates of the $B$ users.
In addition, we assume $D=o(B\ell/\log\ell)$.

\item[(PND)] The probabilistic nonuniform deletion model. Here, we assume that each of the $\ell$ coordinates of the $B$ users is {\em independently} edited with probability $p$, where $p$ may depend on $\ell$. In addition, we assume that the edit probability is relatively small, 
i.e., that $p=o(1/\log \ell)$.
Hence, the total number of edits in a data block is $o(\ell/\log\ell)$.
\end{enumerate}

\begin{theorem} [Scheme T]\label{prop:expectationT}
Consider Scheme T. Let $C$ be the communication cost between a user and a storage node, with edits occurring 
in accordance with the (UD), (CND) and (PND) models.
Then
\[E[C]=\eta(\ell)\, \ell\log q,\] where
\begin{enumerate}
\item For model (UD), we have
 \begin{numcases}
{\lim_{\ell\to\infty} \eta(\ell)}
= \frac{B-1}{B+1}, & \mbox{if $B$ is constant}, \label{eq:TAa}\\
\ge 1-2e^{-c}, & \mbox{if $\lim_{\ell\to\infty}B(\ell)/\ell=c$}.\label{eq:TAb}
\end{numcases}
\item For model (CND), we have
\be
\lim_{\ell\to\infty} \eta(\ell)\ge \frac{D}{D+1}. \label{eq:TB}
\ee
\item For model (PND), we have 
\be{\lim_{\ell\to\infty} \eta(\ell)}
\ge 1-e^{-c},  \mbox{\quad if $\lim_{\ell\to\infty}B(\ell)\, p(\ell)=c$}.\label{eq:TB2}
\ee
\end{enumerate}
\end{theorem}

\begin{proof}

Recall the definitions of $\imax$, $\imin$ and $I$ from Section \ref{sec:traditional}. 
For the model (UD), we have $|I|=\imax -\imin$, while $|I|=\ell-\imin$ holds 
for models (CND) and (PND). Hence, $\eta(\ell)\ell=E[|I|]=E[\imax]-E[\imin]=\ell- 2E[\imin]$ for model (UD), 
and $\eta(\ell)\ell=\ell-E[\imin]$ for models (CND) and (PND).
Therefore, our problem reduces to estimating $E[\imin]$ for the various models.
To do so, we follow the standard derivation methods in order statistics (for example, see David and Nagaraja \cite{DavidNagaraja:1970}), and point out that the range (span, support) of the observed deletion positions is studied in order statistics under the name \emph{sample range}. 

Recall that $E[\imin]=\sum_{i=1}^\ell \prob(\imin\ge i)$ and
observe that 
\begin{equation*}
\prob(\imin\ge i)=
\begin{cases}
\left(\frac{\ell-i+1}{\ell}\right)^B, &\mbox{for model (UD)},\\
\frac{\binom{B(\ell-i+1)}{D}}{\binom{B\ell}{D}}, &\mbox{for model (CND)}\\
(1-p)^{B(i-1)}, &\mbox{for model (PND)}.\\
\end{cases}
\end{equation*}

Hence, for model (UD), 
we have $E[\imin]/\ell=\sum_{i=1}^\ell i^B/\ell^{B+1}$. 
When $B$ is constant, $E[\imin]/\ell\to 1/(B+1)$ as $\ell\to \infty$
yielding \eqref{eq:TAa}. When $\lim_{\ell\to\infty} B/\ell=\kappa$,
we have
\[
\frac{E[\imin]}{\ell}\le \frac 1\ell+(\ell-1)\frac{(\ell-1)^B}{\ell^{B+1}}=\frac 1\ell+\left(1-\frac 1\ell\right)^{B+1}\to e^{-\kappa},
\]
\noindent yielding \eqref{eq:TAb}.

On the other hand, for model (CND),  we have%
\footnote{A similar formulation of $E[\imin]$ can be found in \cite{balakrishnan2003bounds}. } 
\begin{equation*}
E[\imin]=\frac{\sum_{i=1}^\ell \binom{Bi}{D}}{\binom{B\ell}{D}} \le \frac{\sum_{i=D}^{B\ell}\binom{i}{D}+(B-1)\binom{B\ell}{D}}{ B\binom{B\ell}{D}}
=\frac{\binom{B\ell+1}{D+1} }{B\binom{B\ell}{D}}+\frac{B-1}{B}
=\frac{B\ell+1}{B(D+1)}+\frac{B-1}{B}
\end{equation*}
As a result, $\lim_{\ell\to\infty}E[\imin]/\ell\le 1/(D+1)$, yielding \eqref{eq:TB}.

Finally, for model (PND), we have
\[
E[\imin]=\sum_{i=1}^\ell (1-p)^{B(i-1)\ell}\le 1+(\ell-1)(1-p)^B.
\]
If $\lim_{\ell\to\infty}B p=c$, then $\lim_{\ell\to\infty} E[\imin]/\ell\le e^{-c}$ yields~\eqref{eq:TB2}.

\end{proof}

\begin{proposition} [Scheme P/V]\label{prop:expectationPV}
Consider either Scheme P or Scheme V.
Define $C$ to be the communication cost between one selected user and a connected storage node,
with edits occurring according to the above models.
Then 
\begin{enumerate}
\item For model (UD),
 \begin{equation*}
E[C]=\log\ell+\log q .
\end{equation*}

\item For model (CND),

\begin{equation*}
E[C]=\frac{D}{B}(\log\ell+\log q) .
\end{equation*}

\item For model (PND), 

%
%
\[ E[C]=p\ell(\log\ell+\log q).\]
\end{enumerate}
\end{proposition}

\begin{proof}
Note that the communication cost between a user and a connected storage node for Schemes P/V 
is given by $d(\log q+\log\ell)$, where $d$ denotes the corresponding number of edits of the user.
Since for the models (UD), (CND) and (PND), the expected number of edits are one, $D/B$ and $p\ell$, respectively, 
the proposition follows after straightforward algebraic manipulations.
\end{proof}

In Appendix \ref{app:reduce}, we explored other methods for reducing the 
expected communication cost of the Schemes P/V. The idea behind the method is to observe the 
location of the edit with largest index, and then communicate all deletion positions based on a
rescaled data block length equal to that index. As an illustration, if the largest index of an edit in a data block of length $\ell$
is of the order $o(\ell)$, then the locations of the deletions would only require $o(\log\,n)$ bits for encoding.
For the uniform and at random deletion models, these savings are unfortunately not significant.

We conclude this analysis by comparing the relationships between the expected communication costs of the various models. Let $C_T$ and $C_P$ be the expected communication cost in Scheme T and Scheme P/V, respectively.
In the following analysis, asymptotics are computed in $\ell$ and we assume that $q$ is either fixed or that $q=O(\ell)$.

We note that under our model assumptions, we have $\lim_{\ell\to\infty} C_P/C_T=0$.
\begin{enumerate}
\item For model (UD), we have $C_T=\Omega(\ell\log q)$, while $C_P=\Theta(\log \ell +\log q)$.
Hence, $\lim_{\ell\to\infty} C_P/C_T=0$.

\item For model (CND), we have $C_T=\Omega(\ell\log q)$ or $C_T\ge \kappa \ell\log q$, for some constant $\kappa$.
Then 
\[\lim_{\ell\to\infty} \frac{C_P}{C_T}
\le \lim_{\ell\to\infty}\frac{D/B(\log\ell+\log q)}{ \kappa \ell\log q}\to 0,\]
since $D=o(B\ell/\log\ell)$.

\item For model (PND), we have $C_T=\Omega(\ell\log q)$ or $C_T\ge \kappa \ell\log q$, for some constant $\kappa$.
Then 
\[\lim_{\ell\to\infty} \frac{C_P}{C_T}
\le \lim_{\ell\to\infty}\frac{p\ell(\log\ell+\log q)}{ \kappa \ell\log q}\to 0,\]
since $p=o(1/\log\ell)$.

\end{enumerate}

Consequently, for all three probabilistic models and our proposed protocols, we achieve a communication complexity that is asymptotically negligible compared to that of the traditional approach given by Scheme T.

\section{Reducing the Expected Communication Cost for Schemes P/V}
\label{app:reduce}

Suppose that the largest edited coordinate in a user data block with $d$ edits equals $\imax$. In this case, edits are confined to a shorter block than the length of the data block, and sending $d\ceiling{\log \imax}$ bits of information for $d$ edits suffices. In what follows, we compute the expected value of  $\ceiling{\log \imax}$ given $d$ edits, in order to determine how much one may save in the communication rate due to this simple observation.

First, observe that for the independent, uniform deletion model, we have
\begin{align*}
E[\ceiling{\log\imax} \mbox{ given $d$ edits}]&=\sum_{i=1}^{\ceiling{\log\ell}} \prob(\ceiling{\log\imax}\ge i)=\sum_{i=1}^{\ceiling{\log\ell}} \prob(\imax\ge 2^{i-1}+1)\\
&=\sum_{i=1}^{\ceiling{\log\ell}} \left(1-\frac{\binom{2^{i-1}}{d}}{\binom{\ell}{d}}\right)=\ceiling{\log\ell}-\sum_{i=1}^{\ceiling{\log\ell}} \frac{\binom{2^{i-1}}{d}}{\binom{\ell}{d}}.
\end{align*}

Let $s(\ell,d)\triangleq\sum_{i=1}^{\ceiling{\log\ell}}{\binom{2^{i-1}}{d}}/{\binom{\ell}{d}}$ be the second term in the previously derived expression, representing the expected number of bits that can be saved
by using the reduced range encoding scheme. It is straightforward to see that 
$\sum_{i=1}^{\ceiling{\log\ell}}{\binom{2^{i-1}}{d}}/{\binom{\ell}{d}}< 
\sum_{i=1}^{\ceiling{\log\ell}} 2^{i-1}/\ell\le (2\ell-1)/\ell<2$.

Assuming model (CND), the expected communication cost per user-storage node pair follows from a generating function approach that gives
\begin{align*}
&\frac{1}{\binom{B+D-1}{D}}[\lambda^D] (1-\lambda)^{-(B-1)}\sum_{d\ge 1}d\lambda^d(\ceiling{\log\ell}-s(\ell,d))\\
&=\frac{\ceiling{\log\ell}}{\binom{B+D-1}{D}}[\lambda^D] (1-\lambda)^{-(B-1)}\sum_{d\ge 1}d\lambda^d 
- \frac{1}{\binom{B+D-1}{D}}[\lambda^D] (1-\lambda)^{-(B-1)}\sum_{d\ge 1}d\lambda^ds(\ell,d)\\
&=\frac{\ceiling{\log\ell}}{\binom{B+D-1}{D}}[\lambda^D] \frac{\lambda}{(1-\lambda)^{B+1}} 
- \frac{1}{\binom{B+D-1}{D}}\sum_{d= 1}^D ds(\ell,d)[\lambda^{D-d}] (1-\lambda)^{-(B-1)}
\end{align*}
\begin{align*}
&= \frac{\binom{B+D-1}{D-1}}{\binom{B+D-1}{D}}\ceiling{\log\ell}
- \frac{1}{\binom{B+D-1}{D}}\sum_{d= 1}^D ds(\ell,d)\binom{B+D-d-2}{D-d}\\
&=\frac{D}{B}\ceiling{\log\ell} 
- \frac{1}{\binom{B+D-1}{D}}\sum_{d= 1}^D \sum_{i=1}^{\ceiling{\log\ell}}\frac{d\binom{2^{i-1}}{d}\binom{B+D-d-2}{D-d}}{\binom{\ell}{d}},
\end{align*}
\noindent where $[x^k]p(x)$ denotes the coefficient of $x^k$ in the polynomial $p(x)$.
\vskip 5pt

Assuming model (PND), the expected communication cost may be evaluated as
\begin{align*}
&\sum_{d= 1}^{\ell} dp^d(1-p)^{\ell-d}\binom{\ell}{d}(\ceiling{\log\ell}-s(\ell,d))=\ceiling{\log\ell}\sum_{d= 1}^{\ell} dp^d(1-p)^{\ell-d}\binom{\ell}{d}-\sum_{d= 1}^{\ell} dp^d(1-p)^{\ell-d}\binom{\ell}{d}s(\ell,d)\\
&=p\ell\ceiling{\log\ell}-\sum_{d= 1}^{\ell} dp^d(1-p)^{\ell-d}\binom{\ell}{d}s(\ell,d)=p\ell\ceiling{\log\ell}
-\sum_{d= 1}^{\ell}\sum_{i=1}^{\ceiling{\log\ell}}{dp^d(1-p)^{\ell-d}\binom{2^{i-1}}{d}}.
\end{align*}

Assuming model (CND) with different values of $\lambda_s$, the expected communication cost for user $s$ is given by 
\begin{align*}
&\frac{\displaystyle\sum_{d_1+d_2+\cdots+d_B=D}\lambda_1^{d_1}\lambda_2^{d_2}\cdots\lambda_B^{d_B}d_s(\ceiling{\log\ell}-s(\ell,d_s))}%
{\displaystyle\sum_{d_1+d_2+\cdots+d_B=D}\lambda_1^{d_1}\lambda^{d_2}\cdots\lambda_B^{d_B}}\\
&=\frac{\displaystyle\sum_{d_s=0}^D\lambda_s^{d_s}d_s(\ceiling{\log\ell}-s(\ell,d_s))\sum_{\sum_{i\ne s}d_i=D-d_s}\lambda_1^{d_1}\lambda_2^{d_2}\cdots\lambda_{s-1}^{d_{s-1}}\lambda_{s+1}^{d_{s+1}}\cdots\lambda_B^{d_B}}%
{\displaystyle\sum_{d_1+d_2+\cdots+d_B=D}\lambda_1^{d_1}\lambda^{d_2}\cdots\lambda_B^{d_B}}\\
&=\frac{\displaystyle\sum_{d_s=0}^D\lambda_s^{d_s}d_s(\ceiling{\log\ell}-s(\ell,d_s))h_{D-d_s}(\lambda_1,\lambda_2,\ldots,\lambda_{s-1},\lambda_{s+1},\ldots,\lambda_{B})}%
{h_{D}(\lambda_1,\lambda_2,\ldots,\lambda_{B})},\\
\end{align*}

\noindent where $h_k(X_1,X_2,\ldots,X_n)$ is the {\em complete homogeneous symmetric polynomial of degree $k$ in $n$ variables}.

If $\lambda_s=\lambda^s$, we can specialize to
\begin{align}
h_{k}(\lambda,\lambda^2,\ldots,\lambda^{B}) 
&=[z^k]\prod_{1\le i\le B} (1-\lambda^iz)^{-1}\label{eq:all}\\
h_{k}(\lambda,\lambda^2,\ldots,\lambda^{s-1},\lambda^{s+1},\ldots,\lambda^{B}) 
&=[z^k]\prod_{1\le i\le B, i\ne s} (1-\lambda^iz)^{-1}\label{eq:allexcepts}.
\end{align}
We remark that \eqref{eq:all} yields the generating function that counts the number of partitions with at most $k$ parts of size at most B,
while \eqref{eq:allexcepts} yields the generating function that counts the number of partitions with at most $k$ parts of size at most B and no part of size $s$.
\end{document}